\documentclass[11pt,letterpaper]{article}

\usepackage[hyphens]{url}  % DO NOT CHANGE THIS
\usepackage{graphicx} % DO NOT CHANGE THIS
\usepackage{caption} % DO NOT CHANGE THIS AND DO NOT ADD ANY OPTIONS TO IT
\usepackage{algorithm}
\usepackage{algorithmic}

\usepackage{newfloat}
\usepackage{listings}
\floatstyle{ruled}
\newfloat{listing}{tb}{lst}{}
\floatname{listing}{Listing}
\usepackage{subcaption}

\usepackage{amsmath,amsfonts,amssymb,amsthm}
\usepackage{mathbbol}
\usepackage{graphicx,color}
\usepackage{boxedminipage}
\usepackage{framed}
\usepackage{thmtools}
\usepackage{thm-restate}
\usepackage{xspace}
\usepackage{vmargin}
\setmarginsrb{2cm}{2cm}{2cm}{2cm}{0pt}{0pt}{0pt}{6mm}
\usepackage{todonotes}
  \usepackage[pdftex, plainpages = false, pdfpagelabels, 
                  bookmarks=false,
                  bookmarksopen = true,
                  bookmarksnumbered = true,
                  breaklinks = true,
                  linktocpage,
                  pagebackref,
                  colorlinks = true,  
                  linkcolor = blue,
                  urlcolor  = blue,
                  citecolor = red,
                  anchorcolor = green,
                  hyperindex = true,
                  hyperfigures
                  ]{hyperref} 
 \usepackage{xifthen}
 \usepackage{tabularx}
 \usepackage{makecell}
\usepackage{arydshln}
 \usetikzlibrary{calc}
 \usepackage{enumitem}
 \usepackage{tcolorbox} 
 
 \usepackage{mathtools}
\usepackage{multirow}

\newcommand{\calc}{\mathcal{C}}

\newcommand{\Oh}{\mathcal{O}}

\newcommand{\cost}{{\sf cost}}

\newcommand{\cut}{{\sf Cut}}

\newcommand{\coord}{{\sf coord}}

\newcommand{\tr}{{\sf tr}}

\DeclareMathOperator{\operatorClassNP}{{\sf NP}}
\newcommand{\classNP}{\ensuremath{\operatorClassNP}}

\DeclareMathOperator{\operatorClassFPT}{{\sf FPT}\xspace}
\newcommand{\classFPT}{\ensuremath{\operatorClassFPT}\xspace}
\DeclareMathOperator{\operatorClassW}{{\sf W}}
\newcommand{\classW}[1]{\ensuremath{\operatorClassW[#1]}}

\newcommand{\bfA}{\mathbf{A}} 
\newcommand{\bfB}{\mathbf{B}} 
\newcommand{\bfC}{\mathbf{C}}

\newcommand{\bfO}{\mathbf{O}} 
 
\newcommand{\bfU}{\mathbf{U}}
\newcommand{\bfV}{\mathbf{V}}  
\newcommand{\bfa}{\mathbf{a}} 
\newcommand{\bfb}{\mathbf{b}} 
\newcommand{\bfc}{\mathbf{c}}

\newcommand{\bfR}{\mathbf{R}} 
\newcommand{\bfS}{\mathbf{S}} 
 
\newcommand{\bfx}{\mathbf{x}} 
\newcommand{\bfX}{\mathbf{X}} 
\newcommand{\bfY}{\mathbf{Y}} 
\newcommand{\bfZ}{\mathbf{Z}} 
\newcommand{\bfz}{\mathbf{z}} 
\newcommand{\bfy}{\mathbf{y}} 
\newcommand{\bfu}{\mathbf{u}} 
\newcommand{\bfv}{\mathbf{v}} 
\newcommand{\bfW}{\mathbf{W}} 
\newcommand{\bfw}{\mathbf{w}} 

\newtheorem{theorem}{Theorem}
\newtheorem{lemma}{Lemma}
\newtheorem{claim}{Claim}[section]
\newtheorem{corollary}{Corollary}

\newtheorem{observation}{Observation}

\newcommand{\pname}{\textsc}
\newcommand{\ProblemFormat}[1]{\pname{#1}}
\newcommand{\ProblemIndex}[1]{\index{problem!\ProblemFormat{#1}}}
\newcommand{\ProblemName}[1]{\ProblemFormat{#1}\ProblemIndex{#1}{}\xspace}

  \newcommand{\probmeanmed}{$k$-means/median}
 \newcommand{\probEMedClust}{\ProblemName{Explainable $k$-median}}
 \newcommand{\probEMeanClust}{\ProblemName{Explainable $k$-means}}  
 \newcommand{\probClustExpl}{\ProblemName{Clustering Explanation}} 
 \newcommand{\probFMedExpl}{\ProblemName{Approximate Explainable $k$-median}} 
 \newcommand{\probFMeanExpl}{\ProblemName{Approximate  Explainable  $k$-means}} 
  \newcommand{\probHS}{\ProblemName{Hitting Set}}



\newlength{\RoundedBoxWidth}
\newsavebox{\GrayRoundedBox}
\newenvironment{GrayBox}[1]%
   {\setlength{\RoundedBoxWidth}{.93\textwidth}
    \def\boxheading{#1}
    \begin{lrbox}{\GrayRoundedBox}
       \begin{minipage}{\RoundedBoxWidth}}%
   {   \end{minipage}
    \end{lrbox}
    \begin{center}
    \begin{tikzpicture}%
       \node(Text)[draw=black!20,fill=white,rounded corners,%
             inner sep=2ex,text width=\RoundedBoxWidth]%
             {\usebox{\GrayRoundedBox}};
        \coordinate(x) at (current bounding box.north west);
        \node [draw=white,rectangle,inner sep=3pt,anchor=north west,fill=white] 
        at ($(x)+(6pt,.75em)$) {\boxheading};
    \end{tikzpicture}
    \end{center}}     

\newenvironment{defproblemx}[2][]{\noindent\ignorespaces%
                                \FrameSep=6pt%
                                \parindent=0pt%
                \vspace*{-1.5em}
                \ifthenelse{\isempty{#1}}{%
                  \begin{GrayBox}{\textsc{#2}}%                
                }{%
                  \begin{GrayBox}{\textsc{#2} parameterized by~{#1}}%  
                }
                \begin{tabular*}{\textwidth}{@{\hspace{.1em}} >{\itshape} p{1.8cm} p{0.8\textwidth} @{}}%        
            }{
                \end{tabular*}%
                \end{GrayBox}%
                \ignorespacesafterend
            }

\title{How to Find a Good Explanation for Clustering?}
\usepackage{authblk}
\author[1]{Sayan Bandyapadhyay}
\author[1]{Fedor V. Fomin}
\author[1]{Petr A. Golovach}
\author[1]{William Lochet}
\author[1]{Nidhi Purohit}
\author[2]{Kirill Simonov}
\affil[1]{Department of Informatics, University of Bergen, Norway}
\affil[2]{Algorithms and Complexity Group, TU Wien, Vienna, Austria}
\affil[ ]{\{sayan.bandyapadhyay, fedor.fomin, petr.golovach, william.lochet, nidhi.purohit\}@uib.no, kirillsimonov@gmail.com}
\date{}
\begin{document}

\maketitle

\begin{abstract}
%\todo[inline]{capitalize the first letters of the words in the title?}

$k$-means and $k$-median clustering are powerful unsupervised machine learning techniques. However, due to complicated dependences on all the features, it is challenging to interpret the resulting cluster assignments. Moshkovitz, Dasgupta, Rashtchian, and  Frost  [ICML 2020] proposed an elegant model of explainable $k$-means and $k$-median clustering. In this model, a  decision tree with $k$ leaves provides a straightforward characterization of the data set into clusters.

 We study two natural algorithmic questions about explainable clustering.  (1) For a given clustering, how to find the ``best explanation'' by using a decision tree with $k$ leaves?  (2) For a given set of points, how to find a decision tree with $k$ leaves minimizing the $k$-means/median objective of the resulting explainable clustering?
To address the first question, we introduce a new model of explainable clustering. Our model, inspired by the notion of outliers in robust statistics, is the following. We are seeking a small number of points (outliers) whose removal makes the existing clustering well-explainable. For addressing the second question, we initiate the study of the model of Moshkovitz et al. from the perspective of multivariate complexity. Our rigorous algorithmic analysis sheds some light on the influence of parameters like the input size, dimension of the data, the number of outliers, the number of clusters, and the approximation ratio,  on the computational complexity of explainable clustering.

\end{abstract}

 %!TEX root = main.tex

\section{Introduction }\label{sec:intro}
Interpretation or explanation of decisions produced by learning models, including clustering, is a significant direction in machine learning (ML) and artificial intelligence (AI), and has given rise to the subfield of Explainable AI. Explainable AI has attracted a lot of attention from the researchers in recent years (see the surveys by Carvalho et al.~\cite{carvalho2019machine} and Marcinkevi{\v{c}}s and Vogt~\cite{marcinkevivcs2020interpretability}). All these works can be divided into two main categories: \emph{pre-modelling} \cite{wang2015falling,ustun2016supersparse,hastie1986generalized,feng2017sparse,lu2018deeppink} and \emph{post-modelling} \cite{ribeiro2016should,shrikumar2017learning,breiman2001random,sundararajan2017axiomatic,lundberg2017unified} explainability. While post-modeling explainability focuses on giving reasoning behind decisions made by black box models, pre-modeling explainability deals with ML systems that are inherently understandable or perceivable by humans. 
 One of the canonical approaches to pre-modelling explainability builds on decision trees~\cite{molnar2020interpretable,murdoch2019interpretable}. In fact, a significant amount of work on explainable clustering is based on unsupervised decision trees \cite{BertsimasOW21,fraiman2013interpretable,geurts2007inferring,ghattas2017clustering,lipton2018mythos,MoshkovitzDRF20}. In each node of the decision tree, the data is partitioned according to some features' threshold value. 
While such a \textit{threshold tree} provides a clear interpretation of the resulting clustering, its cost measured by the standard \probmeanmed{} objective can be significantly worse than the cost of the optimal clustering. 
Thus,  on the one hand, the efficient algorithms developed for  \probmeanmed{} clustering~\cite{Aggarwalbook13}  are often challenging to explain. On the other hand, the easily explainable models could output very costly clusterings. Subsequently, Moshkovitz et al. \cite{MoshkovitzDRF20}, in a fundamental work, posed the natural algorithmic question of whether it is possible to kill two birds with one stone? To be precise, is it possible to design an efficient procedure for clustering that 
\begin{itemize}
\item Is explainable by a small decision tree; and
\item  Does not cost significantly more than  the cost of an optimal \probmeanmed{} clustering?
\end{itemize}

 To address this question,  Moshkovitz et al. \cite{MoshkovitzDRF20} introduced   explainable \probmeanmed{} clustering. In this scheme, a clustering is represented by a binary (\textit{threshold}) tree whose leaves correspond to clusters, and each internal node corresponds to partitioning a collection of points by a threshold on a fixed coordinate. Thus, the number of leaves in such a tree is $k$, the number of clusters sought. Also, any cluster assignment can be explained by the thresholds along the corresponding root-leaf path.  
  %We postpone the formal definition of explainable clustering till the next section and use Fig.~\ref{fig:fig} as an example.  
  %The left side, 
  For example, consider Fig.~\ref{fig:fig}: Fig.~\ref{fig:sfig1}  shows an optimal $5$-means clustering of a 2D data set; Fig.~\ref{fig:sfig2} shows an explainable $5$-means clustering of the same data set; The threshold tree inducing the explainable clustering is shown in Fig.~\ref{fig:sfig3}.   The tree has five leaves, corresponding to 5 clusters. Note that in this model of explainability, any clustering has a clear geometric interpretation, where each cluster is formed by a set of axis-aligned cuts defined by the tree. 
 As Moshkovitz et al. argue, the classical $k$-means clustering algorithm leads to more complicated clusters while the threshold tree leads to an easy explanation. The advantage of the explainable approach becomes even more evident in higher dimensions when many feature values in $k$-means contribute to the formation of the clusters. 
 
\begin{figure}[h]
%  \label{fig:1}
    \centering
  \begin{subfigure}[h]{.3\columnwidth}
    \centering
    \includegraphics[width=.95\linewidth]{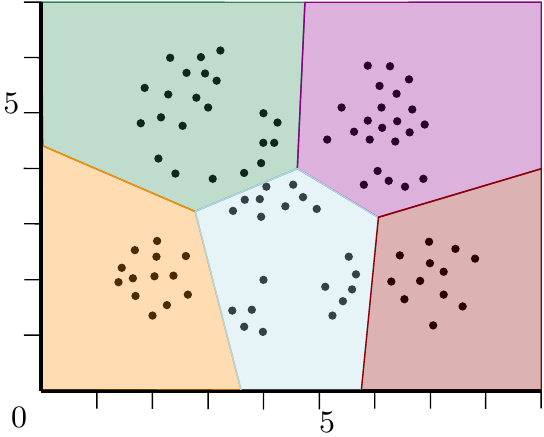}
    \subcaption{}
    \label{fig:sfig1}
  \end{subfigure}\hfill
  \begin{subfigure}[h]{.3\columnwidth}
    \centering
    \includegraphics[width=.95\linewidth]{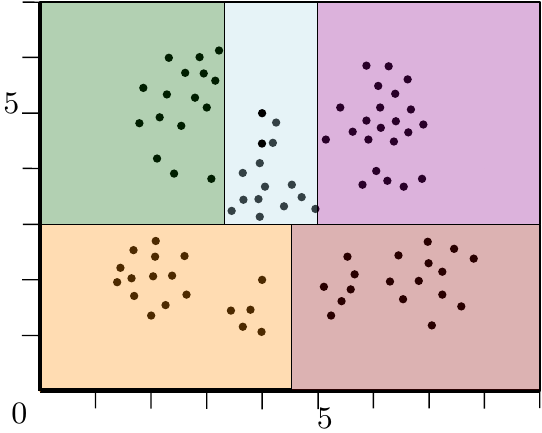}
    \subcaption{}
    \label{fig:sfig2}
  \end{subfigure}\hfill
  \begin{subfigure}[h]{.3\columnwidth}
    \centering
    \includegraphics[width=.95\linewidth]{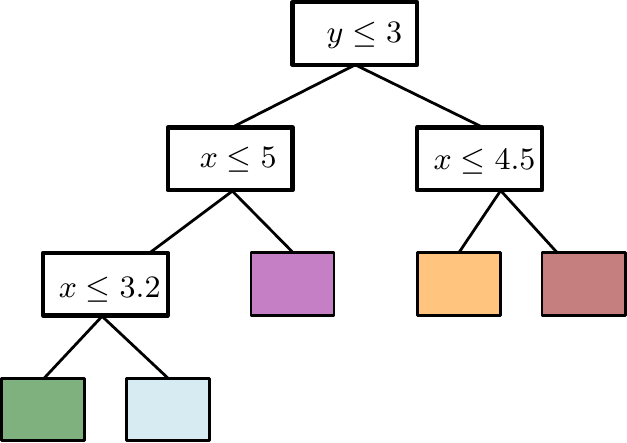}
    \subcaption{}
  \label{fig:sfig3}
  \end{subfigure}
  \caption{(a) An example  of an optimal solution to 5-means. (b) An explainable $5$-means clustering and (c) the corresponding threshold tree.}
  \label{fig:fig}
  \end{figure}

 Moshkovitz et al. \cite{MoshkovitzDRF20} define the quality of any explainable clustering as the ``cost of explainability'', that is the ratio of the cost of the explainable clustering to the cost of an optimal clustering.  Subsequently, they obtain efficient algorithms for computing explainable clusterings whose ``cost of explainability'' is $\Oh(k)$
for $k$-median and $\Oh(k^2)$ for $k$-means. They also show that the this ratio is at least $\Omega(\log k)$ in both cases. Recently, a series of works has been dedicated to improving these bounds.
In the low-dimensional setting, Laber and Murtinho~\cite{LaberM21} showed an upper bound of $\Oh(d \log k)$ and $\Oh(dk \log k)$ for $k$-median and $k$-means respectively.
In general, Makarychev and Shan~\cite{MakarychevS21}, Gamlath, Jia, Polak, and Svensson~\cite{GamlathJPS21} Esfandiari, Mirrokni, and Narayanan~\cite{EsfandiariMN21} showed independently a $k (\log k)^{\Oh(1)}$ upper bound for $k$-means and a $(\log k)^{\Oh(1)}$ upper bound for $k$-median, while also improving the lower bound for $k$-means to $\Omega(k)$. For low dimensions this was improved by Charikar and Hu~\cite{CharikarH21}, who showed an upper bound of $k^{1 - 2/d} (\log k)^{\Oh(1)}$ for $k$-means.

\medskip\noindent\textbf{Our contributions.} In this work, we propose a new model for explaining a clustering, called \probClustExpl. 
Our approach to explainability is inspired by the research on robustness in statistics and machine learning, especially the vast field of outlier detection and removal in the context of clustering \cite{chen2008constant,10.1145/3301446,feng2019improvedFPT,charikar2001algorithms,ChakrabartyGK16,harris2019lottery,krishnaswamy2018constant}. 
 In this model, we are given a \probmeanmed{} clustering and we would like to explain the clustering by a threshold tree \emph{after removing a subset of points}. To be precise, we are interested in finding a subset of points $S$ (which are to be removed) and a threshold tree $T$ such that the explainable clustering induced by the leaves of $T$ is exactly the same as the given clustering after removing the points in $S$. For the given clustering, we define an optimal (or {best}) explainable clustering to be the one that minimizes the size of $S$, i.e. for which the given clustering can be explained by removing the minimum number of points. Thus in  
  \probClustExpl  we measure the ``explainability'' as the number of outlying points whose removal turns the given clustering into  an explainable clustering. The reasoning behind the new measure of cluster explainability is the following.  In certain situations, we would be satisfied with a small decision tree explaining  clustering of all but a few outlying data points.  We note that for a given clustering that is already an explainable clustering, i.e. can be explained by a threshold tree, the size of $S$ is 0.  

 In  Fig.~\ref{fig:fig1}, we provide an example of an optimal  $5$-means clustering of exactly the same data set as  in Fig.~\ref{fig:fig}. However, the new explainable clustering is obtained in a different way. If we remove a small number of points (in  Fig.~\ref{fig:outliers} these are the 9 red larger points), then the explainable clustering is same as the optimal clustering after removing those 9 points.  
 
\begin{figure}[h]
  \label{fig:1}
  \centering

  \begin{subfigure}[t]{.47\columnwidth}
    \centering
    \includegraphics[width=.85\linewidth]{figure_1.pdf}
    \subcaption{}
    \label{fig:opt1}
  \end{subfigure}\hfill
  \begin{subfigure}[t]{.47\columnwidth}
    \centering
    \includegraphics[width=.85\linewidth]{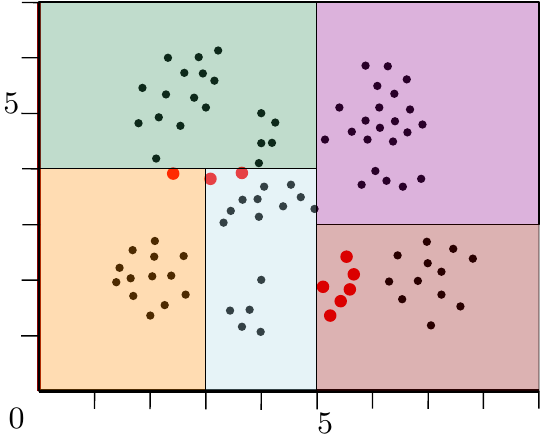}
    \subcaption{}
    \label{fig:outliers}
  \end{subfigure}%
  \caption{(a) An optimal $5$-clustering and (b) an explainable clustering that fits this clustering after removing the larger (red) points. }
  \label{fig:fig1}
  \end{figure}

We note that \probClustExpl corresponds to the classical machine learning setting of interpreting a black-box model, i.e. it lies within the scope of post-modeling explainability. Surprisingly, this area is widely unexplored when it comes to rigorous algorithmic analysis of clustering explanation. Consequently, we study \probClustExpl from the perspective of computational complexity. Our new model naturally raises the following algorithmic questions: (i) \emph{Given a clustering, how efficiently can one decide whether  the clustering can be explained by a threshold tree (without removing any points)?} and (ii) \emph{Given a clustering and an integer $s$, how efficiently can one decide whether the clustering can be explained by removing $s$ points?} 
%We provide several positive and negative answers to this generic question.
% 
%  \begin{tcolorbox}
%Given a clustering, how efficiently one can decide whether it is an explainable clustering? Given a clustering and an integer $s$, how efficiently one can decide whether the clustering can be explained by removing $s$ points? 
% \end{tcolorbox} 

In our work, we design a polynomial time algorithm that resolves the first question. 
Regarding the second question, we give an algorithm that in time 
$2^{2\min\{s,k\}}\cdot n^{2d}\cdot (dn)^{\Oh(1)}$ decides whether a given  clustering of $n$ points in $\mathbb{R}^d$ could be 
 explained by removing $s$ points. 
 %As a corollary, our algorithms implies that recognizing, whether a given clustering is explainable, could be done in polynomial time.  
 We also give an $n^{\Oh(1)}$ time  $(k-1)$-approximation algorithm for \probClustExpl. That is, 
 we give a polynomial time algorithm that returns a solution set of at most $s(k-1)$ points that are to be removed, whereas any  best explainable clustering removes $s$ points. Moreover, we provide an efficient data reduction procedure that reduces an instance of   \probClustExpl  to an equivalent instance with at most $r=2(s+1)dk$ points in $\mathbb{R}^d$ with integer coordinates within the range  $\{1,\ldots, r\}$. The procedure can be used to speed up \emph{any} algorithm for \probClustExpl, as long as $n > 2(s + 1) dk$.
 We complement our algorithms by showing a hardness lower bound. In particular, we show that \probClustExpl cannot be approximated within a factor of $F(s)$ in time $f(s) (nd)^{o(s)}$, for any functions $F$ and $f$, unless Exponential Time Hypothesis (ETH)~\cite{ImpagliazzoPZ01} fails. All these results appear in Section \ref{sec:explanation}.

We also provide new insight into the computational complexity of the model of Moshkovitz et al. \cite{MoshkovitzDRF20}.  While   the vanilla $k$-median and $k$-means problems are NP-hard  
for $k=2$ \cite{AloiseDHP09,DrineasFKVV04,dasgupta2008hardness} or $d=2$ \cite{MahajanNV12},  this is not the case for explainable clustering! We design two simple algorithms  computing optimal (best) explainable clustering with \probmeanmed{} objective that run in time $(4nd)^{k+\Oh(1)}$ and $n^{2d}\cdot n^{\Oh(1)}$, respectively. Hence  for constant $k$ or constant $d$, an optimal explainable clustering can be computed in polynomial time. 
The research on approximation algorithms on the ``cost of explainability'' in~\cite{MoshkovitzDRF20,CharikarH21,EsfandiariMN21,GamlathJPS21,LaberM21,MakarychevS21} implicitly assumes that solving the problem exactly is NP-hard. However,   we did not find a  proof of this fact in the literature.  To fill this gap, we obtain the following  hardness lower bound: An optimal explainable clustering cannot be found in $f(k)\cdot n^{o(k)}$ time for any computable function $f(\cdot)$, unless Exponential Time Hypothesis (ETH) fails. This lower bound demonstrates that asymptotically the running times of our simple algorithms are unlikely to be improved. 
Our reduction also yields that the problem is NP-hard.
These results are described in Section \ref{sec:explainable}.

Finally, we combine the above two explainability models 
%It is possible to obtain better algorithms with dependence on $k$ and $d$ if we relax the requirements on the quality of the solution. 
 to obtain the Approximate Explainable Clustering model: For a collection of $n$ points in $\mathbb{R}^d$ and a positive real constant $\varepsilon<1$, we seek whether we can
   identify at most $\varepsilon n$ outliers, such that the cost of explainable \probmeanmed{} of the remaining points does not exceed the optimal cost of an explainable \probmeanmed{} clustering of the original data set. Thus, if we are allowed to remove a small number of points, can we do as good as any original optimal solution? 
   While our hardness result of Section~\ref{sec:explainable} holds for explaining the whole dataset, by ``sacrificing'' a small fraction of points it might be possible to solve the problem more efficiently.
   %The problem  \probFMedExpl is defined similarly. 
 And indeed, for this model, we obtain an algorithm whose running time $(\frac{8dk}{\epsilon})^k\cdot n^{\Oh(1)}$ has a significantly better dependence on $d$ and $k$. For example, compare this with the above time bounds of $(4nd)^{k+\Oh(1)}$ and $n^{2d}\cdot (dn)^{\Oh(1)}$. This algorithm appears in Section \ref{sec:fractional}. See Table \ref{table:1} for a summary of all our results. 

\renewcommand{\arraystretch}{1.5}
\begin{center}
\begin{table}[t]
\centering
\begin{tabular}{ |c|c|c| }
%\begin{tabular}{ |p{3cm}|p{4cm}|p{3.4cm}| }
 \hline
% \diagbox[height=8ex, width=3.56em]{\raisebox{0.3\height}{\enspace $S$}}{ \raisebox{-1\height}{\ $D$}}
 %$S$\textbackslash $D$ 
 \centering \textbf{Model} & \textbf{Algorithms/Upper bounds} & \textbf{Hardness/Lower bounds} \\ 
 \hline
 \centering Clustering Explanation  & \makecell{$2^{2\min\{s,k\}}n^{2d}n^{\Oh(1)}$\\$(k-1)$-approximation\\Reduction to $\Oh(sdk)$ points }& No $F(s)$-approx. in $f(s) (nd)^{o(s)}$\\ 
\hline
\centering Explainable Clustering  & \centering $(4nd)^{k+\Oh(1)}$ $n^{2d}\cdot n^{\Oh(1)}$ & $f(k)\cdot n^{o(k)}$ \\
\hline
\centering \makecell{Approximate\\Explainable Clustering} & \centering $(\frac{8dk}{\epsilon})^k\cdot n^{\Oh(1)}$ & \\
\hline
% %\multirow{2}{*}{1} &  &1 & \multirow{2}{*}{$poly(n,k)$} & \multirow{2}{*}{Gibson~{\em et.~al}~\cite{GibsonKKPV12}} \\
% %& $\mathbb{R}^d$ $(l_2)$ &1+$\eps$ & &\\
% \multirow{2}{*}{$V$} & $\np$-hard$^{1}$ & \multirow{2}{*}{Polytime} & \multirow{2}{*}{$\ptas$~\cite{BandyapadhyayR17}}\\[1ex]
% & {$\ptas$~\cite{BandyapadhyayR17}} & (trivial) &\\[1ex]
% \hline
% \multirow{3}{*}{$H\sqcup V$} & \multirow{2}{*}{5-approx.} &  \multirow{2}{*}{$\np$-hard } & {$\np$-hard} \\ [1ex]
% &  &  & 7-approx.\\ [1ex]
% & $\ptas$  & {2-approx.} & $(3+\epsilon)$-approx. \\ [1ex]
%  \hline
\end{tabular}
\caption{A summary of our results.}
\label{table:1}
\end{table}
\end{center}

% The main difference between \probEMeanClust  and  \probClustExpl  is that we seek a low-cost clustering that could provide a reasonable explanation in the first one. The second corresponds to the classical machine learning setting of interpreting a black-box model. Surprisingly, this area is widely unexplored when it comes to the rigorous algorithmic analysis of clustering explanation. This brings us to the second main question of our work.

%  While \cite{MoshkovitzDRF20} defines the quality of explainable clustering as the ``cost of explainability'', that is the ratio of the cost of the explainable clustering and the cost of an optimal clustering, to study   \probClustExpl
%  we propose a new measure of explainability. 
%  %  
%The second  algorithmic question, strongly related to the explainable clustering model of Moshkovitz et al. \cite{MoshkovitzDRF20} 
%concerns the classical machine learning  setting of interpreting the black box  model.    \begin{tcolorbox}[colback=green!5!white,colframe=blue!40!black]
% How to find the best explanation of a $k$-clustering efficiently? 
% \end{tcolorbox}
%In other words,  given a $k$-clustering, how  to find efficiently an explainable clustering that is its  ``best fit''?
%
%
%
%\todo[inline]{Short informal definition of explainable clustering, maybe one more  figure?}

%\input{ourwork}

%\input{otherwork}

%\todo[inline]{FF: Some unconsistance with bold letters. Suggestion, do not use bold letters at all. }

%!TEX root = neurips_2021.tex
\section{Preliminaries}\label{sec:prelim} 

% \paragraph{Vectors and clusters.} 
% In our paper, we consider collections $\bfX$  of points (vectors) of $\mathbb{R}^d$. We underline that  some points of a collection may be identical. 
% However, to simplify notation, we assume throughout the paper that identical points of  $\bfX$ are distinct elements of $\bfX$ assuming that the points are supplied with unique identifiers. By this convention, we often refer to (sub)collections of points as (sub)sets and apply the standard set notation.  

\paragraph{\probmeanmed.}
Given a collection $\bfX=\{\bfx_1,\ldots,\bfx_n\}$ of $n$ points in $\mathbb{R}^d$ and a positive integer $k$, the task of \emph{$k$-clustering} is to partition $\bfX$ into $k$ parts $\bfC_1,\ldots,\bfC_k$, called \emph{clusters}, such that the \emph{cost} of clustering is minimized. We follow the convention in the previous work  \cite{MoshkovitzDRF20} for defining the cost. 
%in such a way that the points in each cluster are close to each other according to a given measure. It is standard to consider the center-based clusterings, like the \emph{$k$-median} and \emph{$k$-means} clusterings, where the quality of a clustering is measured as the ``closeness'' of the points to ``centers'' of the clusters.  
In particular, for $k$-means, we consider the Euclidean distance and for $k$-median, the Manhattan distance.
 For a collection of points $\bfX'$ of $\mathbb{R}^d$, we define 
\begin{equation}\label{eq:mean}
\cost_2(\bfX')=\min_{\bfc\in\mathbb{R}^d}\sum_{\bfx\in \bfX'}\|\bfc-\bfx\|_2^2,
\end{equation}
and call the point $\bfc^*\in \mathbb{R}^d$ minimizing the sum in  (\ref{eq:mean}) the \emph{mean} of  $\bfX'$.
For a clustering $\{\bfC_1,\ldots,\bfC_k\}$ of $\bfX\subseteq \mathbb{R}^d$, its \emph{$k$-means} (or simply means) cost is 
%  \begin{equation*}
$\cost_2(\bfC_1,\ldots,\bfC_k)=\sum_{i=1}^k  \cost_2(\bfC_i)$.
% \end{equation*} 
With respect to the Manhattan distance, we define analogously $\cost_1(\bfX')=\min_{\bfc\in\mathbb{R}^d}\sum_{\bfx\in \bfX'}\|\bfc-\bfx\|_1$, which is minimized at the \emph{median} of $\bfX'$,  and $\cost_1(\bfC_1,\ldots,\bfC_k) = \sum_{i=1}^k  \cost_1(\bfC_i)$, which we call the \emph{$k$-median} (or simply median) cost of the clustering.
%The task of the well-known \probMeanClust problem is to find a  $k$-clustering minimizing the mean cost, and \probMedClust is defined similarly with respect to the median cost. 

%Recall, that for a given collection of points $\bfX\subseteq \mathbb{R}^d$, its median and mean can be easily computed in polynomial time. In particular, this implies the following observation.

% \begin{observation}\label{obs:com-cost}
% For a collection  $\bfX\subseteq \mathbb{R}^d$, $\cost_1(\bfX)$ and $\cost_2(\bfX)$ can be computed in polynomial time. 
% \end{observation}

%However, 
%these problems are hard and, moreover, the structure of clusters produced by various algorithm solving these problem may be quite complicated and difficult to interpret. It is often very desirable to have  a simple clustering procedure explaining why a point belongs to a specific cluster. A nice and transparent approach is based on using binary decision trees. We follow the scheme proposed by  Moshkovitz et al.~\cite{MoshkovitzDRF20}.  In this scheme, a clustering is represented by a binary tree whose leaves correspond to clusters, and each internal node corresponds to a partitioning a collection of points by thresholding on a single coordinate. 

%Clearly, $|\coord_i(\bfX)|\leq n$ for every $i\in\{1,\ldots,n\}$.
% For two vectors $\bfx,\bfy\in\mathbb{R}^d$, we write $\bfx\leq \bfy$  ($\bfx< \bfy$, respectively) to denote that $\bfx[i]\leq \bfy[i]$ ($\bfx[i]< \bfy[i]$, respectively) for every $i\in\{1,\ldots,d\}$. We highlight that when we write $\bfx<\bfy$, we require the strict  inequality for \emph{every} coordinate. 
 
\paragraph{Explainable clustering.} For a vector $\bfx\in\mathbb{R}^d$, we use $\bfx[i]$ to denote the $i$-th element (coordinate) of the vector for $i\in\{1,\ldots,d\}$. 
%A   transparent approach to explainable  clustering is based on binary decision trees. We follow the scheme proposed by  Moshkovitz et al.~\cite{MoshkovitzDRF20}.  
Let $\bfX$ be a collection of points of $\mathbb{R}^d$. For $i
\in\{1,\ldots,d\}$ and $\theta\in\mathbb{R}$, we define $\cut_{i,\theta}(\bfX)=(\bfX_1,\bfX_2)$, where $\{\bfX_1,\bfX_2\}$ is a partition of $\bfX$ with 
\begin{equation*} 
\bfX_1=\{\bfx\in \bfX\mid \bfx[i]\leq \theta\} \text{ and }\bfX_2=\{\bfx\in \bfX\mid \bfx[i]>\theta\}.
\end{equation*}
Then, given a collection $\bfX\subseteq \mathbb{R}^d$ and a positive integer $k$, we cluster $\bfX$ as follows. 
If $k=1$, then $\bfX$ is the unique cluster.  If $k=2$, then we choose $i\in \{1,
\ldots,d\}$ and $\theta\in\mathbb{R}$ and construct two clusters $\bfC_1$ and $
\bfC_2$, where $(\bfC_1,\bfC_2)=\cut_{i,\theta}(\bfX)$. For $k>2$, we select $i\in \{1,
\ldots,d\}$ and $\theta\in\mathbb{R}$, and construct a partition $(\bfX_1,\bfX_2)=\cut_{i,\theta}(\bfX)$ of $\bfX$. 
Then clustering of $\bfX$ is defined recursively as the union of a $k_1$-clustering of $\bfX_1$ and a $k_2$-clustering of $\bfX_2$ for some 
  integers $k_1$ and $k_2$ such that $k_1+k_2=k$.
We say that a clustering $\{\bfC_1,\ldots,\bfC_k\}$ is an \emph{explainable $k$-clustering} of a collection of points $\bfX\subseteq \mathbb{R}^d$ if  $\bfC_1,\ldots,\bfC_k$ can be constructed by the described procedure.

%  We also consider the approximate variant of the problem. 
%  %This model captures a clustering that is easily explainable for the larger part of the data points. 
% %Further, we ask the question whether we can get an explainable clustering not for the whole set of point but for a part. To this direction, we introduce the \probFMedExpl and  \probFMeanExpl problems. 
%  In \probFMeanExpl, we are given a collection of $n$ points $\bfX\subseteq \mathbb{R}^d$, a positive integer $k\leq n$, and a positive real constant $\varepsilon<1$. Then the task is to find a collection of points $\bfY\subseteq \bfX$ with $|\bfY|\geq (1-\varepsilon)|\bfX|$ and an explainable $k$-clustering of $\bfY$ whose $k$-median cost does not exceed the optimum $k$-median cost of an explainable $k$-clustering for the original collection of points $\bfX$. Note that we ask about the construction of $\bfY$ and the corresponding clustering as the decision variant is trivial. Observe also that the optimum cost is unknown a priori. 
%  \probFMedExpl differs only by the clustering measure.
%  %In these problems, we are given just a collection of points, but we also can assume that some clustering is already constructed and we have to explain it. 

\paragraph{Threshold tree.} It is useful to represent an explainable $k$-clustering as 
%Such a clustering can be represented as 
a triple $(T,k,\varphi)$, called a \emph{threshold tree}, where $T$ is a rooted binary tree with $k$ leaves, where each nonleaf node has two children called \emph{left} and \emph{right}, respectively, and $\varphi\colon U\rightarrow \{1,
\ldots,d\}\times \mathbb{R}$, where $U$ is the set of nonleaf nodes of $T$. For each node $v$ of $T$, we compute a collection of points $\bfX_v\subseteq \bfX$. For the root $r$, $\bfX_r=\bfX$.   
Let $v$ be a nonleaf node of $T$ and let $u$ and $w$ be its left and right children, respectively,
%with children $u$ and $w$, 
and assume that $\bfX_v$ is constructed.  We compute  $(\bfX_u,\bfX_w)=\cut_{\varphi(v)}(\bfX)$. If $v$ is a leaf, then $\bfX_v$ is a cluster. 
A clustering $\{\bfC_1,\ldots,\bfC_k\}$ is an explainable $k$-clustering of a collection of points $\bfX\subseteq \mathbb{R}^d$ if there is a threshold tree $(T,k,\varphi)$ such that $\bfC_1,\ldots,\bfC_k$ are the clusters corresponding to the leaves of $T$. Note that $T$ is a full binary tree with $k$ leaves and the total number of such trees is the $(k-1)$-th Catalan number, which is upper bounded by $4^k$.

For a collection $\bfX=\{\bfx_1,\ldots,\bfx_n\}$ of $n$ points and $i\in\{1,\ldots,d\}$, we denote by $\coord_i(\bfX)$ the set of distinct  values of $i$-th coordinates $\bfx_j[i]$ for $j\in\{1,\ldots,n\}$.  
%Clearly, $|\coord_i(\bfX)|\leq n$ for every $i\in\{1,\ldots,n\}$. 
It is easy to observe that in the construction of a threshold tree for a set of points $\bfX\subseteq \mathbb{R}^d$, it is sufficient to consider cuts $\cut_{i,\theta}$ with $\theta\in\coord_i(\bfX)$; we call such values of $\theta$ and cuts  \emph{canonical}. We say that a threshold tree $(T,k,\varphi)$ for a collection of points $\bfX\subseteq \mathbb{R}^d$ is \emph{canonical}, if for every nonleaf node $u\in V(T)$, $\varphi(u)=(i,\theta)$ where $\theta\in \coord_i(\bfX)$. 
Throughout the paper we consider only canonical threshold trees. 

\paragraph{Parameterized complexity and ETH.} 
A \emph{parameterized problem} $\Pi$ is a subset of $\Sigma^*\times \mathbb{N}$, where $\Sigma$ is a finite alphabet. Thus, an instance of $\Pi$ is a pair $(I,k)$, where $I\subseteq\Sigma^*$ and $k$ is a nonnegative integer called a \emph{parameter}. It is said that a parameterized problem $\Pi$ is \emph{fixed-parameter tractable} (\classFPT) if it can be solved in $f(k)\cdot |I|^{\Oh(1)}$ time for some computable function $f(\cdot)$. 
The parameterized complexity theory also provides tools to refute the existence of an \classFPT algorithm for a parameterized problem. The standard way is to show that the considered problem is hard in the parameterized complexity classes \classW{1} or \classW{2}. We refer to the book~\cite{CyganFKLMPPS15} for the formal definitions of the parameterized complexity classes.  The basic complexity assumption  of the theory is that for the class  \classFPT, formed by all parameterized fixed-parameter tractable problems, $\classFPT\subset \classW{1}\subset \classW{2}$. The hardness is proved by demonstrating a parameterized reduction from a problem known to be hard in  the considered complexity class.  
A \emph{parameterized reduction} is a many-one reduction that takes an input $(I,k)$ of the first problem, and in $f(k)|I|^{\Oh(1)}$ time outputs an equivalent instance $(I',k')$ of the second problem with $k'\leq g(k)$, where $f(\cdot)$ and $g(\cdot)$ are computable functions. Another way to obtain lower bounds is to use the \emph{Exponential Time Hypothesis (ETH)} formulated by Impagliazzo, Paturi and Zane~\cite{ImpagliazzoP99,ImpagliazzoPZ01}. For an integer $k\geq 3$, let $q_k$ be the infimum of the real numbers $c$ such that the \textsc{$k$-Satisfiability} problem can be solved in time $\Oh(2^{c n})$, where $n$ is the number of variables. Exponential Time Hypothesis states that $\delta_3>3$. In particular, ETH implies that  \textsc{$k$-Satisfiability} cannot be solved in time $2^{o(n)}n^{\Oh(1)}$.

%!TEX root = neurips_2021.tex
\section{Clustering Explanation}\label{sec:explanation}
%TO BE WRITTEN:

%- POLYNOMIAL GREEDY ALGORTHIM, also as an approximation algorithm (Sayan)

%- W[2]-HARDNESS / Inapproximability? (Kirill)

\paragraph{Clustering explanation.}  In  the \probClustExpl problem, the input contains a $k$-clustering $\{\bfC_1,\ldots,\bfC_k\}$ of $\bfX\subseteq \mathbb{R}^d$ and a nonnegative integer $s$, and the task is to decide whether there is a collection of points $W\subseteq \bfX$ with $|W|\leq s$ such that $\{\bfC_1\setminus W,\ldots,\bfC_k\setminus W\}$ is an explainable $k$-clustering. Note that some $\bfC_i\setminus W$ may be empty here.  

\subsection{A Polynomial-time $(k-1)$-Approximation}

In the optimization version of \probClustExpl, we are given a $k$-clustering $\calc=\{\bfC_1,\ldots,\bfC_k\}$ of $\bfX$ in $\mathbb{R}^d$, and the goal is to find a minimum-sized subset $W\subseteq \bfX$ such that $\{\bfC_1\setminus W,\ldots,\bfC_k\setminus W\}$ is an explainable clustering. In the following, we design an approximation algorithm for this problem based on a greedy scheme. 

For any subset $W\subseteq \bfX$, let $\calc-W=\{\bfC_1\setminus W,\ldots,\bfC_k\setminus W\}$. Also, for any subset $Y \subseteq \bfX$, define the clustering induced by $Y$ as $\calc(Y)= \{\bfC_1\cap Y,\ldots,\bfC_k\cap Y\}$. Denote by OPT$(Y)$ the size of the minimum-sized subset $W$ such that the clustering $\calc(Y)-W$ is explainable. First, we have the following simple observation which follows trivially from the definition of OPT$(.)$.  

\begin{observation}\label{obs:subsetcheaper}
 For any subset $Y \subseteq \bfX$, \emph{OPT}$(Y)\le $ \emph{OPT}$(\bfX)$. 
\end{observation}

For any cut $(i,\theta)$ where $i\in \{1,
\ldots,d\}$ and $\theta\in\coord_i(\bfX)$, let $L(i,\theta)=\{\bfx\in \mathbb{R}^d\mid \bfx[i]\le \theta\}$ and  $R(i,\theta)=\{\bfx\in \mathbb{R}^d\mid \bfx[i]> \theta\}$. 
%Now, consider any $Y\subseteq \bfx$. We define a cost function for each cluster in $\calc(Y)$ w.r.t. $(i,\theta)$ in the following way. For any cluster $\bfC\in \calc(Y)$ if $\bfC \subset L(i,\theta)$ or $\bfC \subset L(i,\theta)$, the cost cost$(\bfC,i,\theta)$ is 0. 
%Otherwise, cost$(\bfC,i,\theta)=\min\{\bfC\cap L(i,\theta), \bfC\cap R(i,\theta)\}$. 

\begin{lemma}\label{lem:cheapcut}
 Consider any subset $Y\subseteq \bfX$ such that $\calc(Y)$ contains at least two non-empty clusters. It is possible to select a cut $(i,\theta)$ for $i\in \{1,
\ldots,d\}$ and $\theta\in \coord_i(Y)$, and a subset $W\subseteq Y$, in polynomial time, such that (i) each cluster in $\calc(Y)-W$ is fully contained in either $L(i,\theta)$ or in  $R(i,\theta)$, (ii) at least one cluster in $\calc(Y)-W$ is in $L(i,\theta)$, (iii) at least one cluster in $\calc(Y)-W$ is in $R(i,\theta)$ and (iv) size of $W$ is at most \emph{OPT}$(Y)$. 
\end{lemma}

Before we prove this lemma, we show how to use it to design the desired approximation algorithm. 

\paragraph{The Algorithm.} We start with the set of all points $\bfX$. We apply the algorithm in Lemma \ref{lem:cheapcut} with $Y=\bfX$ to find a cut $(i,\theta)$  and a subset $W_1\subseteq \bfX$ such that each cluster in $\calc(\bfX)-W_1$ is fully contained in either $L(i,\theta)$ or in  $R(i,\theta)$. Let $\bfX_1=(\bfX\setminus W_1)\cap L(i,\theta)$ and $\bfX_2=(\bfX\setminus W_1)\cap R(i,\theta)$. We recursively apply the above step on both $\bfX_1$ and $\bfX_2$ separately. If at some level the point set is a subset of a single cluster, we simply return.  

The correctness of the above algorithm trivially follows from Lemma \ref{lem:cheapcut}. In particular, the recursion tree of the algorithm gives rise to the desired threshold tree. Also, the algorithm runs in polynomial time, as each successful cut $(i,\theta)$ can be found in polynomial time and the algorithm  finds only $k-1$ such cuts that separate the clusters. The last claim follows due to the properties (ii) and (iii) in Lemma \ref{lem:cheapcut}. 

Consider the threshold tree generated by the algorithm. For each internal node $u$, let $X_u$ be the corresponding points and $W_u$ be the points removed from $X_u$ for finding an explainable clustering of the points in $X_u\setminus W_u$. Note that we have at most $k-1$ such nodes. The total number of points removed from $\bfX$ for finding the explainable clustering is $\sum_u |W_u|$. By Lemma \ref{lem:cheapcut}, \[|W_u|\le  \text{OPT}(X_u).\] Now, as $X_u \subseteq \bfX$, by Observation \ref{obs:subsetcheaper}, OPT$(X_u)\le $ OPT$(\bfX)$. It follows that \[\sum_u |W_u|\le (k-1)\cdot \text{OPT}(\bfX). \]

\begin{theorem}\label{thm:k-appr}
 There is a polynomial-time $(k-1)$-approximation algorithm for the optimization version of \probClustExpl. 
\end{theorem}

By noting that OPT$(\bfX)=0$ if $\calc$ is an explainable clustering, we obtain the following corollary.  

\begin{corollary}\label{cor:expl}
 Explainability of any given $k$-clustering in $\mathbb{R}^d$ can be tested in polynomial time. 
\end{corollary}

\begin{proof}[Proof of Lemma \ref{lem:cheapcut}]
 We probe all possible choices for cuts $(i,\theta)$ with  $i\in \{1,
\ldots,d\}$ and $\theta\in \coord_i(Y)$, and select one which incurs the minimum cost. We also select a subset $W$ of points to be removed w.r.t. each cut. The cost of such a cut is exactly the size of $W$. 

Fix a cut $(i,\theta)$. We have the following three cases. In the first case, for all clusters in $\calc(Y)$, strictly more than half of the points are contained in $L(i,\theta)$. In this case select a cluster $\bfC$ which has the minimum intersection with $L(i,\theta)$. Put all the points in $\bfC\cap L(i,\theta)$ into $W$. Also, for any other cluster $\bfC' \in \calc(Y)$, put the points in $\bfC'\cap R(i,\theta)$ into $W$. The second case is symmetric to the first one -- for all clusters in $\calc(Y)$, strictly more than half of the points are contained in $R(i,\theta)$. In this case we again select a cluster $\bfC$ which has the minimum intersection with $R(i,\theta)$. Put all the points in $\bfC\cap R(i,\theta)$ into $W$. Also, for any other cluster $\bfC' \in \calc(Y)$, put the points in $\bfC'\cap L(i,\theta)$ into $W$. In both of the above cases, the first three desired properties are satisfied for $\calc(Y)-W$. In the third case, for each cluster $\bfC \in \calc(Y)$, add the smaller part among $\bfC\cap L(i,\theta)$ and $\bfC\cap R(i,\theta)$ to $W$. In case $|\bfC\cap L(i,\theta)|=|\bfC\cap R(i,\theta)|$, we break the tie in a way so that properties (ii) and (iii) are satisfied. As $\calc(Y)$ contains at least two clusters this can always be done. Moreover, property (i) is trivially satisfied.  

In the above we showed that for all the choices of the cuts, it is possible to select $W$ so that the first three properties are satisfied. Let $w_m$ be the minimum size of the set $W$ over all cuts. As we select a cut for which the size of $W$ is minimized, it is sufficient to show that $w_m \le $ OPT$(Y)$. 

Let $k'$ be the number of clusters in $\calc(Y)$. Consider any optimal set $W^*$ for $Y$ such that $\calc(Y)-W^*$ is explainable. Let $(i^*,\theta^*)$ be the canonical cut corresponding to the root of the threshold tree corresponding to the explainable clustering $\calc(Y)-W^*$. Such a cut exists, as $\calc(Y)$ contains at least two clusters. Let $\widehat{W}$ be the set selected in our algorithm corresponding to the cut $(i^*,\theta^*)$. In the first of the above mentioned three cases, suppose $W^*$ does not contain the part $\bfC\cap L(i^*,\theta^*)$ fully for any of the $k'$ clusters $\bfC\in \calc(Y)$. In other words, $\calc(Y)-W^*$ contains points from each such part $\bfC\cap L(i^*,\theta^*)$. But, then even after choosing the root cut $(i^*,\theta^*)$ we still need $k'$ more cuts to separate the points in $(Y\setminus W^*)\cap L(i^*,\theta^*)$, which contains points from all the $k'$ clusters. However, by definition, the threshold tree must use only $k'$ cuts and hence we reach to a contradiction. Hence, $\bfC^*\cap L(i^*,\theta^*)$  must be fully contained in $W^*$ for some $\bfC^*\in \calc(Y)$. In this case,  our algorithm adds the points in $\bfC\cap L(i^*,\theta^*)$ to $\widehat{W}$ such that the size $|\bfC\cap L(i^*,\theta^*)|$ is minimized over all $\bfC\in \calc(Y)$ and for any other cluster $\bfC' \in \calc(Y)$, we put the points in $\bfC'\cap R(i^*,\theta^*)$ into $\widehat{W}$. Thus, $|\widehat{W}| \le |W^*|=$ OPT$(Y)$. The proof for the second case is the same as the one for the first case. We discuss the proof for the third case. Consider the  clusters $\bfC\in \calc(Y)$ such that both $\bfC\cap L(i^*,\theta^*)$ and $\bfC\cap R(i^*,\theta^*)$ are non-empty. Note that these are the only clusters whose points are put into $\widehat{W}$. But, then $W^*$ must contain all the points from at least one of the parts $\bfC\cap L(i^*,\theta^*)$ and $\bfC\cap R(i^*,\theta^*)$. For each such cluster $\bfC$, we add the smaller part among $\bfC\cap L(i,\theta)$ and $\bfC\cap R(i,\theta)$ to $\widehat{W}$. Hence, in this case also $|\widehat{W}| \le |W^*|=$ OPT$(Y)$. The lemma follows by noting that $w_m\le |\widehat{W}|$.        
\end{proof}

%Next, we describe the other results on \probClustExpl. 

\subsection{Exact algorithm}

% Here we discuss a series of exact algorithmic results for  \probClustExpl. First, we observe that Corollary~\ref{cor:expl} immediately gives the following corollary. 

% \begin{corollaryO}\label{cor:expl2}
%  \probClustExpl can be solved in $n^{s+\Oh(1)}$ time.
%  \end{corollaryO}

% Similarly to Observation~\ref{obs:naive}, we can solve \probClustExpl by a recursive algorithm.

% \begin{propositionO}\label{prop:naive-expl}
% \probClustExpl can be solved in $2^{\Oh(k\log k)}\cdot (nd)^k\cdot n^{\Oh(1)}$ time.
% \end{propositionO}
 
 Our $2^{2\min\{s,k\}}\cdot n^{2d}\cdot (dn)^{\Oh(1)}$  time algorithm is based on a novel dynamic programming scheme. Here, we briefly describe the algorithm. 
%  We can also extend the dynamic programming algorithm from Proposition~\ref{prop:DP} to \probClustExpl. The main complication compared to \probEMeanClust is the following.
Our first observation is that each subproblem can be defined w.r.t. a bounding box in $\mathbb{R}^d$, as each cut used to split a point set in any threshold tree is an axis-parallel hyperplane. The number of such distinct bounding boxes is at most $n^{2d}$, as in each dimension a box is specified by two bounding values. This explains the $n^{2d}$ factor in the running time. Now, consider a fixed bounding box corresponding to a subproblem containing a number of given clusters, may be partially.    
If a new canonical cut splits a cluster, then one of the two resulting parts has to be removed, and this choice has to be passed on along the dynamic programming. As we remove at most $s$ points and the number of clusters is at most $k$, the number of such distinct choices can be bounded by $2^{2\min\{s,k\}}$. This roughly gives us the following theorem. 
%We show that, when performed carefully, this information can be incorporated in the dynamic programming state at the cost of an additional $2^{2\min\{s,k\}}$ factor. 
%The detailed proof is quite technical and follows next.% Appendix \ref{sec:exact-alg}.

\begin{theorem}\label{thm:DP-expl}
\probClustExpl  can be solved in $2^{2\min\{s,k\}}\cdot n^{2d}\cdot (dn)^{\Oh(1)}$  time.
\end{theorem}

Before we move to the formal proof of the theorem, let us introduce some specific notations.
Let $\bfa,\bfb\in \mathbb{R}^d$ be such that $\bfa< \bfb$. We denote $(\bfa,\bfb]=\{\bfx\in\mathbb{R}^d\mid \bfa<\bfx\leq \bfb\}$ and call $(\bfa,\bfb]$ an \emph{interval}. For a collection of points $X\subseteq \mathbb{R}^d$, we say that $X$ is \emph{in} $(\bfa,\bfb]$ if $X\subseteq (\bfa,\bfb]$, $X$ is \emph{outside} $(\bfa,\bfb]$ if $X\cap (\bfa,\bfb]=\emptyset$,   
  and we say that $(\bfa,\bfb]$ \emph{splits} $X$ if $X\cap(\bfa,\bfb]\neq\emptyset$ and $X\setminus(\bfa,\bfb]\neq\emptyset$.  
% For $X\subseteq \mathbb{R}^d$, the \emph{truncation} of $X$ with respect to $(\bfa,\bfb]$, is $\tr_{(\bfa,\bfb]}=X\cap (\bfa,\bfb]$. 

Let $\mathcal{X}$ be a family of disjoint collections of points of $\mathbb{R}^d$.
A subfamily $\mathcal{Y}\subseteq\mathcal{X}$ is said to be \emph{$(\bfa,\bfb]$-proper} if %$\mathcal{Y}\subseteq \mathcal{X}$ such that 
(i) every $X\in\mathcal{X}$ that is in $(\bfa,\bfb]$ is in $\mathcal{Y}$, and (ii) 
every $X\in\mathcal{X}$ that is outside $(\bfa,\bfb]$ is not included in $\mathcal{Y}$. Note that $X\subseteq\mathcal{X}$ that are split by $(\bfa,\bfb]$ may be either in $\mathcal{Y}$ or not in $\mathcal{Y}$.
 %For a family of disjoint collections of points $\mathcal{X}$ of $\mathbb{R}^d$, 
 The  \emph{truncation} of $\mathcal{X}$ with respect to $(\bfa,\bfb]$,
is the family  
$$\tr_{(\bfa,\bfb]}(\mathcal{X})=\{X\cap(\bfa,\bfb]\mid X\in\mathcal{X}\text{ s.t. }X\cap(\bfa,\bfb]\neq\emptyset\}.$$
% For $\bfx\in \mathbb{R}^d$ and a collection $X\subseteq \mathbb{R}^d$, we say that $\bfx$  \emph{splits} $X$ if there are $\bfy,\bfz\in X$ such that $\bfy\leq\bfx$ and $\bfx<\bfz$.   
 %A subfamily $\mathcal{Y}\subseteq\mathcal{X}$ is said to be \emph{$(\bfa,\bfb]$-proper} if $\mathcal{Y}\subseteq \tr_{(\bfa,\bfb]}(\mathcal{X})$ such that every $X\in\mathcal{X}$ that is in $(\bfa,\bfb]$ is in $\mathcal{Y}$.
 %, i.e., $\tr_{(\bfa,\bfb]}(\mathcal{X})\setminus \mathcal{Y}$ may contains only collections in $\mathcal{X}$ that are split by the interval.
% For $\bfx\in \mathbb{R}^d$ and a collection $X\subseteq \mathbb{R}^d$, we say that $\bfx$  \emph{splits} $X$ if there are $\bfy,\bfz\in X$ such that $\bfy\leq\bfx$ and $\bfx<\bfz$.      
%
%Let   $\bfa,\bfb\in \mathbb{R}^d$ for $\bfa<\bfb$. Let also $\mathcal{X}$ be a family of disjoint collections of points $\mathbb{R}^d$. For an integer $s\geq 0$, a subfamily $\mathcal{Y}\subseteq\mathcal{X}$ is \emph{$s$-feasible} if (i) every $X\in\mathcal{X}$ such that $X$ is in $(\bfa,\bfb]$ is included in $Y$, (ii) every$X\in\mathcal{X}$ such that $X$ is outside $(\bfa,\bfb]$ is not included in $Y$, and (iii) $(\bfa,\bfb]$ splits at most $s$ collections in $\mathcal{Y}$.  
 %
  %Let $\bfa,\bfb\in \mathbb{R}^d$ for $\bfa<\bfb$. Let also $\mathcal{X}$ be a family of disjoint collections of points $\mathbb{R}^d$. 
  For an integer $s\geq 0$, 
  $\mathcal{X}$ is \emph{$s$-feasible}  with respect to $(\bfa,\bfb]$ if $(\bfa,\bfb]$ splits at most $s$ collections in $\mathcal{X}$.  
 
%   \begin{theorem}\label{thm:DP-expl}
%  \probClustExpl  can be solved in $2^{2\min\{s,k\}}\cdot n^{2d}\cdot (dn)^{\Oh(1)}$  time.
%   \end{theorem}
   
  \begin{proof}[Proof of Theorem~\ref{thm:DP-expl}]
 Let $(\mathcal{C},s)$ be an instance of \probClustExpl, where $\mathcal{C}=\{\bfC_1,\ldots,\bfC_k\}$ for disjoint collections of points $\bfC_i$ of $\mathbf{R}^d$. Let $\bfX=\bigcup_{i=1}^k\bfC_i$.
 Following the proof of Proposition~\ref{prop:DP}, we say that a vector $\bfz\in(\mathbb{R}\cup\{\pm\infty\})^d$ is \emph{canonical}  if $\bfz[i]\in \coord_i(\bfX)\cup\{\pm\infty\}$ for every $i\in\{1,\ldots,d\}$. 
 
 For every pair of canonical vectors $(\bfa,\bfb)$ such that $\bfa<\bfb$ and $\mathcal{C}$ is $s$-feasible with respect to $(\bfa,\bfb]$, and every %nonempty 
 $(\bfa,\bfb]$-proper $\mathcal{S}=\{\bfS_1,\ldots,\bfS_\ell\}\subseteq\mathcal{C}$, 
 we denote by $\omega(\bfa,\bfb,\mathcal{S})$ the minimum size of a collection of points $\bfW\subseteq \bfX\cap(\bfa,\bfb]$ such that $\{\bfS_1'\setminus \bfW,\ldots,\bfS_\ell'\setminus \bfW\}$, where $\{\bfS_1',\ldots,\bfS_\ell'\}=\tr_{(\bfa,\bfb]}(\mathcal{S})$,  is an explainable  $\ell$-clustering. We assume that $\omega(\bfa,\bfb,\mathcal{S})=0$ if $\mathcal{S}$ is empty.
 % such that for every $i\in\{i,\ldots,\ell\}$, $\bfS_i'\setminus W=\bfX\cap(\bfa_i,\bfb_i]$ with $\mathcal{C}$ being $s$-feasible with respect to $(\bfa_i,\bfb_i]$.
 %, if such a set $W$ of size at most $s$ exists. 
 %We denote this value by $\omega(\bfa,\bfb,\mathcal{S})$. 
 We compute 
 \begin{equation}\label{eq:def-w}
 w(\bfa,\bfb,\mathcal{S})=\omega(\bfa,\bfb,\mathcal{S}) +\sum_{\bfC_i\in \mathcal{S}}|\bfC_i\setminus (\bfa,\bfb]|+\sum_{\bfC_i\in\mathcal{C}\setminus\mathcal{S}}|\bfC_i\cap(\bfa,\bfb]|.
 \end{equation}
 Since we are interested only in clustering that can be obtained by deleting at most $s$ points, we assume that  $\omega(\bfa,\bfb,\mathcal{S})=+\infty$ if this value is bigger than $s$. This slightly informal agreement simplifies arguments.  In particular, observe the two sums in (\ref{eq:def-w}) give the value that is bigger than $s$ if $\mathcal{S}$ is not $s$-feasible with respect to $(\bfa,\bfb]$. In fact, this is the reason why these sums are included in (\ref{eq:def-w}).  
%Note also that  $\{\bfS_1'\setminus W,\ldots,\bfS_\ell'\setminus W\}$ may include empty sets. To accommodate this special case, we formally extend the definition of $\omega(\bfa,\bfb,\mathcal{S})$ for arbitrary subfamilies $\mathcal{S}$ of $\mathca{C}$ by setting  $\omega(\bfa,\bfb,\mathcal{S})=0$ if  $(\bfa,\bfb]=\emptyset$ (i.e. $\bfa\not<\bfb$).
 
  Notice that $(\mathcal{C},s)$ is a yes-instance of \probClustExpl if and only if $w(\bfa^*,\bfb^*,\mathcal{C})\leq s$, where $\bfa^*[i]=-\infty$ and $\bfb^*[i]=+\infty$ for $i\in\{1,\ldots,d\}$. 
 
 The values $w(\bfa,\bfb,\mathcal{S})$ are computed depending on $\ell=|\mathcal{S}|$.  If $\ell=0$, that is, $\mathcal{S}=\emptyset$, then $\omega(\bfa,\bfb,\mathcal{S})=0$ and
 $w(\bfa,\bfb,\mathcal{S})=\sum\limits_{\bfC_j\in\mathcal{C}}|\bfC_j\cap(\bfa,\bfb]|$. 
 If $\ell =1$, then $\omega(\bfa,\bfb,\mathcal{S})=0$ by definition. Then  $\mathcal{S}=\{\bfC_i\}$ for some $i\in\{1,\ldots,k\}$ such that $\bfC_i\cap(\bfa,\bfb]\neq\emptyset$ and 
 $$w(\bfa,\bfb,\mathcal{S})=|\bfC_i\setminus (\bfa,\bfb]|+\sum_{\bfC_j\in\mathcal{C}\setminus\{\bfC_i\}}|\bfC_j\cap(\bfa,\bfb]|.$$
 
 Assume that $\ell\geq 2$, and the  values of   $\omega(\bfa',\bfb',\mathcal{S}')$ 
 are computed for $|\mathcal{S}'|<\ell$. 
 
 For $i\in \{1,\ldots,d\}$ and $\theta\in \coord_i(\bfX)$ such that $\bfa[i]< \theta<\bfb[i]$, we define the vectors $\bfa^{i,\theta}$ and $\bfb^{i,\theta}$ by setting
 $$\bfa^{i,\theta}[j]=
 \begin{cases}
 \theta&\mbox{if }j=i,\\
 \bfa[j]&\mbox{if }j\neq i,
 \end{cases}
 \text{ and }
\bfb^{i,\theta}[j]=
 \begin{cases}
 \theta&\mbox{if }j=i,\\
 \bfb[j]&\mbox{if }j\neq i.
 \end{cases}
$$
 We also say that $(i,\theta)$ is \emph{$s$-feasible} if $\mathcal{C}$ is $s$-feasible with respect to $(\bfa,\bfb^{i,\theta}]$ and $(\bfa^{i,\theta},
 \bfb]$.
For an $s$-feasible $(i,\theta)$, a partition  $(\mathcal{S}_1,\mathcal{S}_2)$ of $\mathcal{S}$ is \emph{$(i,\theta)$-proper} if $\mathcal{S}_1$ and $\mathcal{S}_2$ are  $(\bfa,\bfb^{i,\theta}]$ and $(\bfa^{i,\theta},
 \bfb]$-proper, respectively. We define 
 \begin{equation*}  
 \delta_{i,\theta}(\mathcal{S})=\sum_{\substack{\bfC_i\in\mathcal{S}\colon \bfC_i\cap(\bfa,\bfb^{i,\theta}]\neq\emptyset\\\text{and }\bfC_i\cap(\bfa^{i,\theta},\bfb]\neq\emptyset}} |\bfC_i\cap(\bfa,\bfb]|. 
 \end{equation*}
 
%For an  $s$-feasible $(i,\theta)$ and an $(i,\theta)$-proper partition $(\mathcal{S}_1,\mathcal{S}_2)$, we define
%$$
%w_1(i,\theta,\mathcal{S}_1,\mathcal{S}_2)=\sum_{S\in \mathcal{S}_2}|S\cap (\bfa,\bfb^{i,\theta}]|\text{ and }w_2(i,\theta,\mathcal{S}_1,\mathcal{S}_2)=\sum_{S\in \mathcal{S}_1}|S\cap (\bfa^{(i,\theta)},\bfb]|.$$  

We compute $\omega(\bfa,\bfb,\mathcal{S})$ by the following recurrence.
\begin{align}\label{eq:rec-dp-expl}
%w(\bfa,\bfb,\mathcal{S})=\min\{&\min_{\bfC_i\in\mathcal{S}}(w(\bfa,\bfb,\mathcal{S}\setminus\{\bfC_i\})+|\bfC_i\setminus(\bfa,\bfb]|),  \nonumber \\
w(\bfa&,\bfb,\mathcal{S})= \min\{(**) + (***)\},
%\min\{\nonumber\\
%&\min\{|\bfC_i\setminus(\bfa,\bfb]|&&+\sum_{\bfC_j\in\mathcal{S}\setminus\{C_i\}}|\bfC_j|\nonumber\\
%& &&+\sum_{\bfC_j\in \mathcal{C}\setminus\mathcal{S} }|\bfC_j\cap(\bfa,\bfb]|\mid\bfC_i\in\mathcal{S} \},\nonumber  \\
 %& \min \{w(\bfa,\bfb^{i,\theta},\mathcal{S}_1)+w(\bfa^{i,\theta},\bfb,\mathcal{S}_2)-\delta_{i,\theta}(\mathcal{S})\mid &1\leq i\leq d,~\theta\in \coord_i(\bfX), \nonumber\\& &(\mathcal{S}_1,\mathcal{S}_2)\text{ is a partition of }\mathcal{S}\text{ s.t.}\nonumber\\
%& &\bfa[i]< \theta<\bfb[i],\nonumber\\
%& & (i,\theta)\text{ is }s\text{-feasible},\nonumber\\
%& &(\mathcal{S}_1,\mathcal{S}_2)\text{ is  }(i,\theta)\text{-proper}\}\}.
\end{align}
%& \min \{w(\bfa,\bfb^{i,\theta},\mathcal{S}_1)+w(\bfa^{i,\theta},\bfb,\mathcal{S}_2)-\delta_{i,\theta}(\mathcal{S})\mid &1\leq i\leq d,~\theta\in \coord_i(\bfX), \nonumber\\
%where the minimum is taken over all $i\in \{1,\ldots,d\}$, $\theta\in \coord_i(\bfX)$ with 
%$\bfa[i]<\theta<\bfb[i]$, and all partitions $(\mathcal{S}_1,\mathcal{S}_2)$ of $\mathcal{S}$ such that (i) $(i,\theta)$ is $s$-feasible and (ii) $(\mathcal{S}_1,\mathcal{S}_2)$ is $(i,\theta)$-proper.  
where the right part is denoted by $(*)$, and
\begin{equation*}
(**)=\min\{|\bfC_i\setminus(\bfa,\bfb]|+\sum_{\bfC_j\in\mathcal{S}\setminus\{C_i\}}|\bfC_j|+\sum_{\bfC_j\in \mathcal{C}\setminus\mathcal{S} }|\bfC_j\cap(\bfa,\bfb]|\mid\bfC_i\in\mathcal{S} \},
\end{equation*}
 \begin{multline*}%\label{eq:rec-dp-two}
(***)=\min \{w(\bfa,\bfb^{i,\theta},\mathcal{S}_1)+w(\bfa^{i,\theta},\bfb,\mathcal{S}_2)-\delta_{i,\theta}(\mathcal{S}) \text{ for } 1\leq i\leq d,~\theta\in \coord_i(\bfX), \nonumber\\
(\mathcal{S}_1,\mathcal{S}_2)\text{ is  partition of }\mathcal{S}\text{ s.t.}, \bfa[i]< \theta<\bfb[i],(i,\theta)\text{ is }s\text{-feasible}, (\mathcal{S}_1,\mathcal{S}_2)\text{ is  }(i,\theta)\text{-proper}\}.
\end{multline*} 
We assume that $(***)=+\infty$ if there is no triple $(i,\theta,(\mathcal{S}_1,\mathcal{S}_2))$ satisfying the conditions in the definition of the set.
We also assume that $(*)=+\infty$
% (and $w(\bfa,\bfb,\mathcal{S})=+\infty$)  if there is no triple $(i,\theta,(\mathcal{S}_1,\mathcal{S}_2))$ satisfying the conditions in the definition of the set in~(\ref{eq:rec-dp-expl}), and we do the same assumption 
if its value proves to be bigger than $s$.

 The correctness of (\ref{eq:rec-dp-expl}) is proved by  showing the inequalities between the left and rights parts in both directions.  
% For this, we denote the right part of (\ref{eq:rec-dp-expl}) by $(*)$ and set
% \begin{align}\label{eq:rec-dp-two}
%(**)=\min \{w(\bfa,\bfb^{i,\theta},\mathcal{S}_1)+w(\bfa^{i,\theta},\bfb,\mathcal{S}_2)-\delta_{i,\theta}(\mathcal{S})\mid &1\leq i\leq d,~\theta\in \coord_i(\bfX), \nonumber\\
%&(\mathcal{S}_1,\mathcal{S}_2)\text{ is  partition of }\mathcal{S}\text{ s.t.}\nonumber\\
%&\bfa[i]< \theta<\bfb[i],\nonumber\\
%& (i,\theta)\text{ is }s\text{-feasible},\nonumber\\
%&(\mathcal{S}_1,\mathcal{S}_2)\text{ is  }(i,\theta)\text{-proper}\}.
%\end{align} 

First, we show that $w(\bfa,\bfb,\mathcal{S})\geq (*)$. This is trivial  if $w(\bfa,\bfb,\mathcal{S})=+\infty$. Assume that this is not the case. Then by our assumption, $w(\bfa,\bfb,\mathcal{S})\leq s$. 
Recall that 
$w(\bfa,\bfb,\mathcal{S})=\omega(\bfa,\bfb,\mathcal{S})+\sum_{\bfC_i\in \mathcal{S}}|\bfC_i\setminus (\bfa,\bfb]|+\sum_{\bfC_i\in\mathcal{C}\setminus\mathcal{S}}|\bfC_i\cap(\bfa,\bfb]|$. 
Let $r=\omega(\bfa,\bfb,\mathcal{S})$ and let $\bfW\subseteq \bfX\cap(\bfa,\bfb]$ be a collection of  $r$ points such that $\mathcal{S}'=\{\bfS_1'\setminus \bfW,\ldots,\bfS_\ell'\setminus \bfW\}$, where $\{\bfS_1',\ldots,\bfS_\ell'\}=\tr_{(\bfa,\bfb]}(\mathcal{S})$,  is an explainable  $\ell$-clustering. Assume that $\mathcal{S}=\{\bfC_{i_1},\ldots,\bfC_{i_\ell}\}$ and $\bfS_j'=\bfC_{i_j}\cap (\bfa,\bfb]$.
 Let $\bfW_i=\bfW\cap \bfS_i'$ for $i\in\{1,\ldots,\ell\}$. 

Notice that it may happen that $\bfW_j=\bfS_j'$ for some $j\in\{1,\ldots,\ell\}$. Then $\bfC_{i_j}\cap(\bfa,\bfb]=\bfW_j$.  
Observe, however, that $\bfC_{i_j}\cap(\bfa,\bfb]=\bfW_j$ for at most $\ell-1$ values of $j$. Suppose that there is $h\in \{1,\ldots,\ell\}$
such that $\bfC_{i_h}\cap(\bfa,\bfb]\neq \bfW_h$ and $\bfC_{i_j}\cap(\bfa,\bfb]=\bfW_j$ for every $h\in\{1,\ldots,\ell\}$ such that $j\neq h$.
In this case, we obtain that 
%$\omega(\bfa,\bfb,\mathcal{S})=\omega(\bfa,\bfb,\mathcal{S}\setminus\{C_{i_j}\})+|\bfC_{i_j}\cap(\bfa,\bfb])$ 
\begin{equation*}
\omega(\bfa,\bfb,\mathcal{S})=\sum_{\bfC_j\in\mathcal{S}\setminus\{C_{i_h}\}}|\bfC_j\cap(\bfa,\bfb]|
\end{equation*}
and 
\begin{equation*}
w(\bfa,\bfb,\mathcal{S})=|\bfC_{i_h}\setminus(\bfa,\bfb]| + \sum_{\bfC_j\in\mathcal{S}\setminus\bfC_{i_h}}|\bfC_j|+\sum_{\bfC_j\in \mathcal{C}\setminus\mathcal{S} }|\bfC_j\cap(\bfa,\bfb]|.
\end{equation*}
%$w(\bfa,\bfb,\mathcal{S})=w(\bfa,\bfb,\mathcal{S}\setminus\{\bfC_{i_j}\})+|\bfC_{i_j}\setminus(\bfa,\bfb]|$. 
Then 
%$w(\bfa,\bfb,\mathcal{S})\geq \min_{\bfC_i\in\mathcal{S}}(w(\bfa,\bfb,\mathcal{S}\setminus\{\bfC_i\})+|\bfC_i\setminus(\bfa,\bfb])$ and 
$w(\bfa,\bfb,\mathcal{S})\geq (**)\geq (*)$.
 Assume from now that this is not the case and $\bfS_j\setminus \bfW_j\neq\emptyset$ for at least two distinct indices $j\in\{1,\ldots,\ell\}$.
Then we show that $w(\bfa,\bfb,\mathcal{S})\geq (***)$.
%By the definition of explainable clusterings, there are $i\in\{1,\ldots,d\}$ and $\theta\in\coord_i(\bfX)$, such that $\bfa[i]\leq\theta\leq\bfb[i]$ and 
%there is a partition $(\hat{\mathcal{S}}_1,\hat{\mathcal{S}}_2)$ of $\mathcal{S}'$ with the property that (i) $\hat{\mathcal{S}}_1$ is an explainable $\ell_1=|\hat{\mathcal{S}}_1|$-clustering with $\bigcup_{S\in\hat{\mathcal{S}}_1}S\subseteq(\bfa,\bfb^{i,\theta}]$, and 
%(ii) $\hat{\mathcal{S}}_2$ is an explainable $\ell_2=|\hat{\mathcal{S}}_2|$-clustering with $\bigcup_{S\in\hat{\mathcal{S}}_2}S\subseteq(\bfa^{i,\theta},\bfb]$. Let 

Because we separate at least two nonempty collections of points, the definition of explainable clustering implies that 
%By the definition of explainable clusterings, 
there are $i\in\{1,\ldots,d\}$ and $\theta\in\coord_i(\bfX)$, such that $\bfa[i]<\theta<\bfb[i]$ and 
there is a partition $(I_1,I_2)$ of $\{1,\ldots,\ell\}$ with the property that (i) $\hat{\mathcal{S}}_1=\{\bfS_i'\setminus \bfW_j\mid j\in I_1 \}$ is an explainable $\ell_1=|I_1|$-clustering with the clusters in $(\bfa,\bfb^{i,\theta}]$, and 
(ii) $\hat{\mathcal{S}}_2=\{\bfS_j'\setminus \bfW_j\mid j\in I_2 \}$ is an explainable $\ell_2=|I_2|$-clustering with the clusters in $(\bfa^{i,\theta},\bfb]$, where both $\hat{\mathcal{S}}_1$ and $\hat{\mathcal{S}}_2$ contain nonempty collections of points. 
Moreover, we assume that if  $\bfS_j\setminus \bfW_j=\emptyset$ for some $j\in\{1,\ldots,\ell\}$, then $\bfS_j\setminus \bfW_j$ is placed in $\hat{\mathcal{S}}_1$ if $\bfS_j$ has a point in $(\bfa,\bfb^{i,\theta}]$ and, otherwise, i.e. if  $\bfS_j$ has only points in $(\bfa^{i,\theta},\bfb]$, it is placed in  $\hat{\mathcal{S}}_2$.

We define $\mathcal{S}_1=\{\bfC_{i_j}\mid j\in I_1\}$ and $\mathcal{S}_2=\{\bfC_{i_j}\mid j\in I_2\}$. For $j\in I_1$, let $\bfW_j^1=(\bfa,\bfb^{i,\theta}]\cap \bfW_j$, and let   $\bfW_j^2=(\bfa^{i,\theta},\bfb]\cap \bfW_j$. 
We set $\bfW^1=\bigcup_{j\in I_1}\bfW_j^1$ and  $\bfW^2= \bigcup_{j\in I_2}\bfW_j^2$. Let also 
$\bfR_1=(\bfa^{i,\theta},\bfb]\cap\big( \bigcup_{j\in I_1}\bfW_j\big)$ and  $\bfR_2=(\bfa,\bfb^{i,\theta}]\cap\big( \bigcup_{j\in I_2}\bfW_j\big)$. Observe that $(\bfW^1,\bfR_1,\bfW^2,\bfR_2)$ is a partition of $\bfW$ where some sets may be empty. 
Denote by $w_1=|\bfW^1|$ and $w_2=|\bfW^2|$, and let $r_1=|\bfR_1|$ and $r_2=|\bfR_2|$. Clearly, $w_1+w_2+r_1+r_2=r$.

Notice that $\mathcal{S}_1'=\{\bfS_j'\setminus \bfR_1\mid j\in I_1\}=\tr_{(\bfa,\bfb^{i,\theta}]}(\mathcal{S}_1)$ and $\mathcal{S}_2'=\{\bfS_j'\setminus \bfR_2\mid j\in I_2\}=\tr_{(\bfa^{i,\theta},\bfb]}(\mathcal{S}_2)$. 
Also $\mathcal{S}_1$ and $\mathcal{S}_2$ are $(\bfa,\bfb^{i,\theta}]$ and $(\bfa^{i,\theta},\bfb]$-proper, respectively. Furthermore, for $h\in\{1,2\}$,
$\hat{\mathcal{S}}_h$ is obtained from $\mathcal{S}_h'$ by deleting the points of $\bfW^1$ from the clusters. 
%This means that $\omega(\bfa,\bfb^{i,\theta},\mathcal{S}_1)\leq w_1$ and  $\omega(\bfa^{i,\theta},\bfb,\mathcal{S}_2)\leq w_2$.
Also we have that 
\begin{equation}\label{eq:in-one}
\sum_{\bfC_j\in \mathcal{S}_1}|\bfC_j\setminus(\bfa,\bfb^{i,\theta}]|=\sum_{\bfC_j\in \mathcal{S}_1}|\bfC_j\setminus(\bfa,\bfb]|+|\bfR_1|,
\end{equation}
\begin{equation}\label{eq:in-two}
\sum_{\bfC_j\in \mathcal{S}_2}|\bfC_j\setminus(\bfa^{i,\theta},\bfb]|=\sum_{\bfC_j\in \mathcal{S}_2}|\bfC_j\setminus(\bfa,\bfb]|+|\bfR_2|,
\end{equation}
and
\begin{equation}\label{eq:out}
\sum_{\bfC_j\in \mathcal{C}\setminus \mathcal{S}_1}|\bfC_j\cap(\bfa,\bfb^{i,\theta}]|+\sum_{\bfC_j\in \mathcal{C}\setminus \mathcal{S}_2}|\bfC_j\cap(\bfa^{i,\theta},\bfb]|=\sum_{\bfC_j\in \mathcal{C}\setminus\mathcal{S}}|\bfC_j\cap(\bfa,\bfb]|+|\bfR_1|+|\bfR_2|.
\end{equation}
%and 
%\begin{equation}\label{eq:out-two}
%\sum_{\bfC_j\in \mathcal{S}_2}|\bfC_j\setminus(\bfa^{i,\theta},\bfb]|=\sum_{\bfC_j\in \mathcal{S}_2}|\bfC_j\setminus(\bfa,\bfb]|+|R_1|.
%\end{equation}
Note also that 
\begin{equation}\label{eq:delta}
\delta_{i,\theta}(\mathcal{S})=|\bfR_1|+|\bfR_2|.
\end{equation}
Then by (\ref{eq:in-one})--(\ref{eq:delta}),
\begin{align}\label{eq:forward}
w(\bfa,\bfb,\mathcal{S})=&(w_1+w_2+r_1+r_2) +\sum_{\bfC_j\in \mathcal{S}}|\bfC_j\setminus (\bfa,\bfb]|+\sum_{\bfC_j\in\mathcal{C}\setminus\mathcal{S}}|\bfC_j\cap(\bfa,\bfb]|    \nonumber\\
=&(w_1+w_2+r_1+r_2)+\sum_{\bfC_j\in \mathcal{S}_1}|\bfC_j\setminus (\bfa,\bfb]|+\sum_{\bfC_j\in \mathcal{S}_2}|\bfC_j\setminus (\bfa,\bfb]|\nonumber\\
 & +\sum_{\bfC_j\in \mathcal{C}\setminus \mathcal{S}_1}|\bfC_j\cap(\bfa,\bfb^{i,\theta}]|  +\sum_{\bfC_j\in \mathcal{C}\setminus \mathcal{S}_2}|\bfC_j\cap(\bfa^{i,\theta},\bfb]|-2r_1-2r_2 \nonumber\\
=& w_1+\sum_{\bfC_j\in \mathcal{S}_1}|\bfC_j\setminus (\bfa,\bfb^{i,\theta}]| +\sum_{\bfC_j\in \mathcal{C}\setminus \mathcal{S}_1}|\bfC_j\cap(\bfa,\bfb^{i,\theta}]|\nonumber\\
&+w_2+\sum_{\bfC_j\in \mathcal{S}_2}|\bfC_j\setminus (\bfa^{i,\theta},\bfb]| +\sum_{\bfC_j\in \mathcal{C}\setminus \mathcal{S}_2}|\bfC_j\cap(\bfa^{i,\theta},\bfb]|-\delta_{i,\theta}(\mathcal{S}).
\end{align}
Recall that $w(\bfa,\bfb,\mathcal{S})\leq s$. 
Note that 
\begin{equation*}\label{eq:feas}
\sum_{\bfC_j\in \mathcal{C}\setminus \mathcal{S}_1}|\bfC_j\cap(\bfa,\bfb^{i,\theta}]|\leq \sum_{\bfC_j\in \mathcal{C}\setminus\mathcal{S}}|\bfC_j\cap(\bfa,\bfb]|+|\bfR_2|.
\end{equation*}
Using (\ref{eq:in-one}), we obtain that 
\begin{equation*} 
\sum_{\bfC_j\in \mathcal{S}_1}|\bfC_j\setminus (\bfa,\bfb^{i,\theta}]|+\sum_{\bfC_j\in \mathcal{C}\setminus \mathcal{S}_1}|\bfC_j\cap(\bfa,\bfb^{i,\theta}]|\leq  w(\bfa,\bfb,\mathcal{S}),
\end{equation*}
which is at most $s$.
This means that $\mathcal{S}_1$ is $(\bfa,\bfb^{i,\theta}]$-proper. Similarly, we have that 
$\mathcal{S}_2$ is  $(\bfa^{i,\theta},\bfb]$-proper.
%In particular, this means that $\sum_{\bfC_i\in \mathcal{S}_1}|\bfC_j\setminus (\bfa,\bfb^{i,\theta}]|\leq s$ and $\sum_{\bfC_j\in \mathcal{S}_2}|\bfC_j\setminus (\bfa^{i,\theta},\bfb]|\leq s$.  Then becase 
%$\mathcal{S}_1$ and $\mathcal{S}_2$ are $(\bfa,\bfb^{i,\theta}]$ and $(\bfa^{i,\theta},\bfb]$-proper, respectively, we obtain that %$\mathcal{S}_1$ and $\mathcal{S}_2$ are 
%$(i,\theta)$
%$s$-feasible. % with respect to $(\bfa,\bfb^{i,\theta}]$ and $(\bfa^{i,\theta},\bfb]$, respectively. 
Therefore, 
\begin{equation*}
w_1+\sum_{\bfC_j\in \mathcal{S}_1}|\bfC_j\setminus (\bfa,\bfb^{i,\theta}]|+\sum_{\bfC_j\in \mathcal{C}\setminus \mathcal{S}_1}|\bfC_j\cap(\bfa,\bfb^{i,\theta}]|\geq w(\bfa,\bfb^{i,\theta},\mathcal{S}_1)
\end{equation*}
and 
\begin{equation*}
w_2+\sum_{\bfC_j\in \mathcal{S}_2}|\bfC_j\setminus (\bfa^{i,\theta},\bfb]|+\sum_{\bfC_j\in \mathcal{C}\setminus \mathcal{S}_2}|\bfC_j\cap(\bfa^{i,\theta},\bfb]|\geq w(\bfa^{i,\theta},\bfb,\mathcal{S}_2).
\end{equation*}

This allows us to extend (\ref{eq:forward}) and conclude that 
\begin{equation*}
w(\bfa,\bfb,\mathcal{S})\geq w(\bfa,\bfb^{i,\theta},\mathcal{S}_1)+w(\bfa^{i,\theta},\bfb,\mathcal{S}_2)\geq (***).
\end{equation*}
This shows that $w(\bfa,\bfb,\mathcal{S})\geq(***)\geq(*)$ and
concludes the proof of the first inequality.

Now we show that $w(\bfa,\bfb,\mathcal{S})\leq (*)$. The inequality is trivial if $(*)=+\infty$.  Suppose that this is not the case. Then, by our assumption about assigning the value to $(*)$, $(*)\leq s$. 

Suppose that the minimum in $(*)$ is achieved in the first part, that is, $(**)\leq (***)$.
%$(*)= w(\bfa,\bfb,\mathcal{S}\setminus\{\bfC_i\})+|\bfC_i\setminus(\bfa,\bfb]|$ for some $\bfC_i\in\mathcal{S}$.
Then 
\begin{equation*}
(*)=|\bfC_i\setminus(\bfa,\bfb]|+\sum_{\bfC_j\in\mathcal{S}\setminus\{C_i\}}|\bfC_j|\nonumber+\sum_{\bfC_j\in \mathcal{C}\setminus\mathcal{S} }|\bfC_j\cap(\bfa,\bfb]|
\end{equation*} 
for some $\bfC_i\in\mathcal{S}$.
We define $\bfW=\bigcup\limits_{\bfC_j\in \mathcal{S}\setminus \{\bfC_\}}|\bfC_j\cap (\bfa,\bfb]|$.
%Let $\hat{\mathcal{S}}=\mathcal{S}\setminus \{\bfC_i\}$ and let $r=\omega(\bfa,\bfb,\hat{\mathcal{S}})$. Then there is a collection  $W\subseteq \bfX\cap(\bfa,\bfb]$ of  $r$ points such that  the family obtained from $tr_{(\bfa,\bfb]}(\hat{\mathcal{S}})$ by the deletion of the points of $W$ is an explainable $(\ell-1)$-clustering. Let $W'=W\cup (\bfC_i\cap (\bfa,\bfb])$. 
We have that the family obtained from $tr_{(\bfa,\bfb]}(\mathcal{S})$ by the deletion of the points of $\bfW$ is an explainable $\ell$-clustering, because we deleted the points in $(\bfa,\bfb]$ of every $\bfC_j\in\mathcal{S}$ excepts $\bfC_i$.
This implies that $\omega(\bfa,\bfb,\mathcal{S})\leq|\bfW|$ and we obtain that 
\begin{multline*} 
w(\bfa,\bfb,\mathcal{S})\leq|\bfW|+\sum_{\bfC_i\in \mathcal{S}}|\bfC_i\setminus (\bfa,\bfb]|+\sum_{\bfC_i\in\mathcal{C}\setminus\mathcal{S}}|\bfC_i\cap(\bfa,\bfb]|\\
=|\bfC_i\setminus(\bfa,\bfb]|+\sum_{\bfC_j\in\mathcal{S}\setminus\{C_i\}}|\bfC_j|\nonumber+\sum_{\bfC_j\in \mathcal{C}\setminus\mathcal{S} }|\bfC_j\cap(\bfa,\bfb]|,
\end{multline*}
which is exactly $(**)$, so $\omega(\bfa,\bfb,\mathcal{S})\leq (*)$.

Assume from now that the minimum in $(*)$ is achieved for the second part, that is, $(*)=(***)<(**)$.
Suppose that 
$i\in\{1,\ldots,d\}$, $\theta\in \coord_i(\bfX)$ where $\bfa[i]\leq \theta\leq \bfb[i]$ and $(i,\theta)$ is $s$-feasible, and a partition $(\mathcal{S}_1,\mathcal{S}_2)$ of $\mathcal{S}$ that is $(i,\theta)$-proper are chosen in such a way that $(**)$ achieves the minimum value for them, that is, $(***)=w(\bfa,\bfb^{i,\theta},\mathcal{S}_1)+w(\bfa^{i,\theta},\bfb,\mathcal{S}_2)-\delta_{i,\theta}(\mathcal{S})$.  
  
Let $\bfR_1=(\bfa^{i,\theta},\bfb]\cap\big(\bigcup_{\bfC_j\in\mathcal{S}_1}\bfC_j)$ and  $\bfR_2=(\bfa,\bfb^{i,\theta}]\cap\big(\bigcup_{\bfC_j\in\mathcal{S}_2}\bfC_j)$. Then we obtain that 
\begin{equation}\label{eq:in-one-b}
\sum_{\bfC_j\in \mathcal{S}_1}|\bfC_j\setminus(\bfa,\bfb^{i,\theta}]|=\sum_{\bfC_j\in \mathcal{S}_1}|\bfC_j\setminus(\bfa,\bfb]|+|\bfR_1|,
\end{equation}
\begin{equation}\label{eq:in-two-b}
\sum_{\bfC_j\in \mathcal{S}_2}|\bfC_j\setminus(\bfa^{i,\theta},\bfb]|=\sum_{\bfC_j\in \mathcal{S}_2}|\bfC_j\setminus(\bfa,\bfb]|+|\bfR_2|,
\end{equation}
\begin{equation}\label{eq:out-b}
\sum_{\bfC_j\in \mathcal{C}\setminus \mathcal{S}_1}|\bfC_j\cap(\bfa,\bfb^{i,\theta}]|+\sum_{\bfC_j\in \mathcal{C}\setminus \mathcal{S}_2}|\bfC_j\cap(\bfa^{i,\theta},\bfb]|
=\sum_{\bfC_j\in \mathcal{C}\setminus\mathcal{S}}|\bfC_j\cap(\bfa,\bfb]|+|\bfR_1|+|\bfR_2|,
\end{equation}
and 
\begin{equation}\label{eq:delta-b}
\delta_{i,\theta}(\mathcal{S})=|\bfR_1|+|\bfR_2|.
\end{equation}

Let $w_1=\omega(\bfa,\bfb^{i,\theta},\mathcal{S}_1)$ and $w_2=\omega(\bfa^{i,\theta},\bfb,\mathcal{S}_2)$. Then there is a collection $\bfW_1\subseteq (\bfa,\bfb^{i,\theta}]\cap\bfX$ such that  the collection of sets obtained from the collections of $\tr_{(\bfa,\bfb^{i,\theta}]}(\mathcal{S}_1)$ by the deletions of the points of $\bfW_1$ is an explainable $|\mathcal{S}_1|$-clustering. Similarly, there is a collection $\bfW_2\subseteq (\bfa^{i,\theta},\bfb]\cap\bfX$ such that  the family of points collections obtained from the collections of $\tr_{(\bfa^{i,\theta},\bfb]}(\mathcal{S}_2)$ by the deletions the points of $\bfW_2$ is an explainable $|\mathcal{S}_2|$-clustering. Consider $\bfW=\bfW_1\cup \bfW_2\cup \bfR_1\cup \bfR_2$. The crucial observation is that the  family  obtained from the collections of $\tr_{(\bfa,\bfb]}(\mathcal{S})$ by the deletions of the points of $\bfW$ is an explainable $\ell$-clustering, where the first cut is $\cut_{i,\theta}$. Using (\ref{eq:in-one-b})--(\ref{eq:delta-b}), we obtain that
\begin{multline*}
(***)=w(\bfa,\bfb^{i,\theta},\mathcal{S}_1)+w(\bfa^{i,\theta},\bfb,\mathcal{S}_2)-\delta_{i,\theta}(\mathcal{S})
=|\bfW|+\sum_{\bfC_i\in \mathcal{S}}|\bfC_i\setminus (\bfa,\bfb]|+\sum_{\bfC_i\in\mathcal{C}\setminus\mathcal{S}}|\bfC_i\cap(\bfa,\bfb]|\\
\geq \omega(\bfa,\bfb,\mathcal{S})+\sum_{\bfC_i\in \mathcal{S}}|\bfC_i\setminus (\bfa,\bfb]|+\sum_{\bfC_i\in\mathcal{C}\setminus\mathcal{S}}|\bfC_i\cap(\bfa,\bfb]|
=w(\bfa,\bfb,\mathcal{S}).
\end{multline*} 
 Since $(*)=(***)$, this completes the correctness proof of the recurrence (\ref{eq:rec-dp-expl}).
 
 In the final stage of the proof, we evaluate the running time. We construct the table of values of $w(\bfa,\bfb,\mathcal{S})$ for pairs $(\bfa,\bfb)$ of canonical vectors such that $\bfa< \bfb$. The total number of such pairs is at most $(n+2)^{2d}$ and they can be constructed in $n^{2d}\cdot (dn)^{\Oh(1)}$ time. We are interested only in $\bfa$ and $\bfb$ such that $\mathcal{C}$ is $s$-feasible with respect to $(\bfa,\bfb]$. Clearly, for given $\bfa$ and $\bfb$, the  $s$-feasibility can be checked in $(dn)^{\Oh(1)}$ time. If $\mathcal{C}$ is $s$-feasible with respect to $(\bfa,\bfb]$, then all $(\bfa,\bfb]$-proper subfamilies $\mathcal{S}$ of $\mathcal{C}$ can be listed by brute force as follows. Observe that 
 $\mathcal{X}=\{\bfC_i\in\mathcal{C}\mid \bfC_i\text{ is in }(\bfa,\bfb]\}$ that can be constructed in polynomial time is a subfamily of every $(\bfa,\bfb]$-proper $\mathcal{S}$. Let $\mathcal{Y}=\{\bfC_i\in \mathcal{C}\mid \bfC_i\text{ is split by }(\bfa,\bfb]\}$.
 Since  $\mathcal{C}$ is $s$-feasible, $|\mathcal{Y}|\leq\min\{s,k\}$. We can construct $\mathcal{Y}$ in polynomial time and then generate at most $2^{\min\{s,k\}}$ subfamilies $\mathcal{Z}$ of $\mathcal{Y}$ in total 
 $2^{\min\{s,k\}}\cdot (dn)^{\Oh(1)}$ time. Then the $(\bfa,\bfb]$-proper subfamilies $\mathcal{S}$ are exactly the families of the form $\mathcal{Z}\cup\mathcal{Y}$.  We obtain that there are at most $2^{\min\{s,k\}}$ the 
 $(\bfa,\bfb]$-proper subfamilies that can be generated in time $2^{\min\{s,k\}}\cdot (dn)^{\Oh(1)}$ time. Then we conclude that the dynamic programming algorithm computes at most $2^{\min\{s,k\}}\cdot n^{2d}$ values of $w(\bfa,\bfb,\mathcal{S})$. 
  
The  value of $w(\bfa,\bfb,\mathcal{S})$ for $|\mathcal{S}|$ is constructed in $(dn)^{\Oh(1)}$ time  if $\ell\leq1$. If $\ell\geq 2$, we are using the recurrence (\ref{eq:rec-dp-expl}). 
Computing  $(**)$ %$\min_{\bfC_i\in\mathcal{S}}(w(\bfa,\bfb,\mathcal{S}\setminus\{\bfC_i\})+|\bfC_i\setminus(\bfa,\bfb]|)$
 can be done in polynomial time.
To compute $(***)$, we  go through $i\in\{1,\ldots,d\}$, $\theta\in\coord_i(\bfX)$, and partitions $(\mathcal{S}_1,\mathcal{S})$ of $\mathcal{S}$, where $(\mathcal{S}_1,\mathcal{S}_2)$ is required to be $(i,\theta)$-proper. This implies that   $(***)$ can be computed in $2^{\min\{s,k\}}\cdot (dn)^{\Oh(1)}$ time. Summarizing, we have that the total running time of the dynamic programming algorithm is    $2^{2\min\{s,k\}}\cdot n^{2d}\cdot (dn)^{\Oh(1)}$. This concludes the proof.    
 \end{proof} 
 
\subsection{Data reduction}\label{sec:kernel}

\begin{restatable}{theorem}{kernel}\label{thm:kern}
    Let $r = 2 (s + 1) dk$. There is a polynomial-time algorithm that, given an instance of \probClustExpl, produces an equivalent one with at most $r$ points in $\{1,\ldots, r\}^d$. 
    %In other words, \probClustExpl admits a polynomial kernel when parameterized by $d + k + s$.
\end{restatable}
  
\begin{proof}
Let $(\mathcal{C},s)$ be an instance of \probClustExpl, where $\mathcal{C}=\{\bfC_1,\ldots,\bfC_k\}$ for disjoint collections of points $\bfC_i$ of $\mathbf{R}^d$. Let $\bfX=\bigcup_{i=1}^k\bfC_i$.
  
Our first aim is to reduce the number of points. For this, we use a procedure that marks essential points. 
  
For every $i\in\{1,\ldots,k\}$ and every $j\in\{1,\ldots,d\}$, do the following:
\begin{itemize}
\item Order the points of $\bfC_i$ by the increase of their $j$-th coordinate; the  ties are broken arbitrarily.
\item \emph{Mark} the first $\min\{s+1,|\bfC_i|\}$ points and the last   $\min\{s+1,|\bfC_i|\}$ points in the ordering.
\end{itemize}
  
The procedure marks at most $2(s+1)dk$ points. Then we delete the remaining unmarked points. Formally, we denote by $\bfY$ the collection of marked points and set $\bfS_i=\bfC_i\cap\bfY$ for all $i\in\{1,\ldots,k\}$.   
Then we consider the instance $(\mathcal{S},s)$ of  \probClustExpl, where $\mathcal{S}=\{\bfS_1,\ldots,\bfS_k\}$. We show the following claim.

\begin{claim}\label{cl:expl}
$(\mathcal{C},s)$ is a yes-instance of \probClustExpl if and only if $(\mathcal{S},s)$ is a yes-instance.
\end{claim}

\begin{proof}[Proof of Claim~\ref{cl:expl}]
Trivially, if   $(\mathcal{C},s)$ is a yes-instance, then $(\mathcal{S},s)$ is a yes-instance, because we just deleted some point to construct $(\mathcal{S},s)$. We show that  
if $(\mathcal{S},s)$ is a yes-instance, then $(\mathcal{C},s)$ is a yes-instance. 

Because $(\mathcal{S},s)$ is a yes-instance, there is a collection of at most $s$ points $\bfW\subseteq \bfY$ such that $\{\bfS_1\setminus \bfW,\ldots,\bfS_k\setminus \bfW\}$ is an explainable $k$-clustering. In other words, there is an explainable clustering of $\bfY\setminus \bfW$ with a canonical threshold tree $(T,k,\varphi)$ such that the clusters $\bfS_1\setminus \bfW,\ldots,\bfS_k\setminus \bfW$ correspond to the leaves of the threshold tree. We claim that if we use the same threshold tree for $\bfX\setminus \bfW$, then $\bfC_1\setminus \bfW,\ldots,\bfC_k\setminus \bfW$ correspond to the leaves. 

The proof is by contradiction. Assume that at least one collections of points corresponding to a leaf is distinct from every $\bfC_1\setminus \bfW,\ldots,\bfC_k\setminus \bfW$. Then there is a node $v\in V(T)$ such that for some $j\in\{1,\ldots,k\}$, $\bfC_j\setminus \bfW$ is split by the cut $\cut_{i,\theta}$ for $(i,\theta)=\varphi(v)$, that is, for $(\bfA,\bfB)=\cut_{i,\theta}(\bfX)$, $\bfA\cap (\bfC_j\setminus \bfW)\neq\emptyset$ and $\bfB\cap (\bfC_j\setminus \bfW)\neq\emptyset$. Observe that either $\bfA\cap (\bfS_j\setminus \bfW)=\emptyset$ or $\bfB\cap (\bfS_j\setminus \bfW)=\emptyset$. We assume without loss of generality that $\bfA\cap (\bfS_j\setminus \bfW)=\emptyset$ (the other case is symmetric). This means that there is an unmarked point $\bfx\in \bfC_j\setminus \bfW$ in $\bfA$ and all the marked points of $\bfC_j\setminus \bfW$ are in $\bfB$. Because $\bfC_j$ has an unmarked point, $|\bfC_j|\geq 2(s+1)+1$.
Following the marking procedure, we order the points of  $\bfC_j$ by the increase of the $i$-th coordinate breaking ties exactly as in the marking procedure. Let $L$ be the collection of the first $s+1$ points that are marked. Since $|\bfW|\leq s$, there is $\bfy\in L\setminus \bfW$. Because $L\setminus \bfW\subseteq \bfS_j\setminus \bfW\subseteq \bfB$, we have that $\bfy[i]>\theta$. Then $\bfx[i]\geq\bfy[i]>\theta$ and $\bfx\in \bfB$; a contradiction.  

We conclude that  if we use $(T,k,\varphi)$ to cluster $\bfX\setminus \bfW$, then $\bfC_1\setminus \bfW,\ldots,\bfC_k\setminus \bfW$ correspond to the leaves. This proves that $(\mathcal{C},s)$ is a yes-instance of 
 \probClustExpl.
 \end{proof}   

We obtained the instance $(\mathcal{S},s)$, where $\bfY=\bigcup_{i=1}^k\bfS_i$ has $\ell\leq 2(s+1)dk$ points, that is equivalent to the original instance.  Now we modify  the points to ensure that they are in $\{1,\ldots,\ell\}^d$. For this, we observe that for each $i\in\{1,\ldots,d\}$, the values of the $i$-th coordinates can be changed if we maintain their order. Formally, we do the following. For every $i\in \{1,
\ldots,d\}$, let $\coord_i(\bfY)=\{\theta_1^i,\ldots,\theta_{r_i}^i\}$, where   $\theta_1^i<\cdots<\theta_{r_i}^i$. For every $\bfy\in\bfY$, we construct a point $\bfz$, by setting
$\bfz[i]=j$, where $\theta_j^i=\bfy[i]$, for each $i\in\{1,\ldots,d\}$. Then for  $\bfS_i$ containing $\bfy$, we replace $\bfy$ by $\bfz$. 
Denote by $\bfZ$ the constructed collection of points, and let $\mathcal{R}=\{\bfR_1,
\ldots,\bfR_k\}$ be the family of the collections of points constructed from $\bfS_1,\ldots,\bfS_k$.

We have that $(\mathcal{R},s)$ is a yes-instance of \probClustExpl if and only if $(\mathcal{S},s)$ is a yes-instance, and $\bfZ\subseteq \{1,\ldots,\ell\}^d$. Then the data reduction algorithm returns $(\mathcal{R},s)$. To complete the proof, it remains to observe that the marking procedure is polynomial, and the coordinates replacement also can be done in polynomial time.
\end{proof}

\subsection{Hardness of approximation}
\label{sec:hardclexpl}

We show that the \probClustExpl problem remains hard when the number of points to delete $s$ is small. Specifically, we provide a parameter-preserving reduction from \textsc{Hitting Set} to \probClustExpl that transfers known results about hardness of approximation for the \textsc{Hitting Set} problem %(Theorem~\ref{thm:hitting_set_hardness})
to \probClustExpl.
Recall that in the \textsc{Hitting Set} problem, the input is a family of sets $\mathcal{A}$ over a universe $U$ together with a parameter $\ell$, and the goal is to decide whether there is a set $H \subset U$ of size at most $\ell$ such that $H$ has non-empty intersection with each set in $\mathcal{A}$. Whenever we consider \classFPT algorithms for \textsc{Hitting Set}, we assume that the problem is parameterized by $\ell$.

The starting point of our reduction is the following result for \textsc{Hitting Set} by Karthik, Laekhanukit, and Manurangsi~\cite{KarthikLM19}.

\begin{theorem}[Theorem 1.4 in~\cite{KarthikLM19}]
    Assuming ETH, no $f(\ell)(|U||\mathcal{A}|)^{o(\ell)}$-time algorithm can approximate \textsc{Hitting Set} to within a factor of $F(\ell)$, for any functions $f$ and $F$ of $\ell$.
    \label{thm:hitting_set_hardness}
\end{theorem}

Note that Theorem 1.4 in~\cite{KarthikLM19} is stated for the \textsc{Dominating Set} problem, however by the standard parameter-preserving reduction from \textsc{Dominating Set} to \textsc{Hitting Set} (see e.g. Theorem 13.28 in \cite{CyganFKLMPPS15}), the statement above immediately follows.

Now, intuitively, given an instance $(U, \mathcal{A}, \ell)$ of \probHS, our reduction constructs clusters $\bfC_0$, \ldots, $\bfC_m$ in $\mathbb{R}^{\sum_{j \in [m]} |A_j|}$. The clusters $\bfC_1$, \ldots, $\bfC_m$ represent the sets in the family $\mathcal{A}=\{A_1,\ldots,A_m\}$, and $\bfC_0$ is a special cluster that needs to be separated from each of $\bfC_1$, \ldots, $\bfC_m$ so that the clustering is explainable. The separation can only be performed by removing special points from $\bfC_0$ each of which corresponds to an element of the universe $U$. Removing such a point allows for separation between $\bfC_0$ and each $\bfC_j$ such that the corresponding set $A_j$ contains the corresponding universe element. The two clusters can be separated along a special coordinate where only that special point ``blocks'' the separation. This is the crux of the reduction, which results in the following theorem.

\begin{theorem}
    For any functions $f$ and $F$, there is no algorithm that approximates \probClustExpl within a factor of $F(s)$ in time $f(s) (nd)^{o(s)}$, unless ETH fails. 
    \label{thm:explanation_hardness}
\end{theorem}

\begin{proof}
    We show a reduction from \textsc{Hitting Set}. Consider an instance of \textsc{Hitting Set} over a universe $U$ with a family of sets $\mathcal{A}$ and the size of the target hitting set $\ell$. We construct the following instance of \probClustExpl, denote $|\mathcal{A}| = m$. The target dimension $d$ is equal to the sum of set sizes in the family $\mathcal{A}$, $d = \sum_{S \in \mathcal{A}} |S|$. The number of clusters in the constructed instance is $m + 1$, and for clarity we denote them by $\bfC_0$, \ldots, $\bfC_m$. The target parameter $s$, that is, the number of points to remove from the cluster, is set exactly to $\ell$.

    Now we describe how the clusters are composed. Intuitively, the clusters $\bfC_1$, \ldots, $\bfC_m$ represent the sets in the family $\mathcal{A}$, and $\bfC_0$ is a special cluster that needs to be separated from each of $\bfC_1$, \ldots, $\bfC_m$ so that the clustering is explainable. The separation can only be performed by removing special points from $\bfC_0$ that each correspond to an element of the universe $U$. Removing such a point allows for separation between $\bfC_0$ and each $\bfC_j$ such that the corresponding set $S_j$ contains the corresponding universe element. The two clusters can be separated along a special coordinate where only that special point ``blocks'' the separation. This is the crux of the reduction.

    Formally, we define point sets $\bfC_0$, \ldots, $\bfC_m$ in terms of coordinates in $\mathbb{R}^d$. The constructed instance is binary, that is, only values zero and one are used in the vectors.
    Order arbitrarily the sets in the family $\mathcal{A}$, i.e. $\mathcal{A} = \{S_1, \ldots, S_m\}$, for each $j \in [m]$, the cluster $\bfC_j$ corresponds to the set $S_j$. The $d$ coordinates are partitioned between the $m$ sets in the following way. For each $j \in [m]$, denote $I_j = \left[\sum_{j' < j} |S_{j'}| + 1, \sum_{j' \le j} |S_{j'}|\right] \subset [d]$. That is, the first range of coordinates $I_1$ are the first $|S_1|$ coordinates, $I_2$ are the following $|S_2|$ coordinates, and so on. Now, for $j \in [m]$, the set $\bfC_j$ consists of $F(\ell) \cdot \ell + 1$ identical points $\bfw_j$, where $\bfw_j[i] = 1$ for $j \in I_j$ and $\bfw_j[i] = 0$ for $j \notin I_j$.
    The set $\bfC_0$ consists of two parts, $\bfC_0 = \bfO \cup \bfV$. First, there are $F(\ell) \cdot \ell + 1$ identical zero vectors in $\bfC_0$, we denote this set by $\bfO$. Second, for each element $u$ in the universe $U$, there is a point $\bfv_u$ in $\bfV$. The coordinates of $\bfv_u$ are set so that for every $j \in [m]$ such that $u \in S_j$, there is exactly one coordinate in $I_j$ where $\bfv_u$ is set to one, and this coordinate is unique for each $u \in S_j$. If for $j \in [m]$, $u \notin S_j$, then $\bfv_u[i] = 0$ for all $i \in I_j$. More specifically, for each $S_j \in \mathcal{A}$, order arbitrarily the elements of $S_j$, $S_j = \{u_1, \ldots, u_{|S_j|}\}$. For $i \in [|S_j|]$, $\bfv_{u_i}[\sum_{j' < j} |S_j| + i] = 1$. After performing the above for each $j \in [m]$, set the remaining coordinates of each vector $\bfv_u$ to zero. This concludes the construction.
    
    Now we show that an $F(s)$-approximate solution to the constructed \probClustExpl instance would imply an $F(\ell)$-approximate solution to the original \textsc{Hitting Set} instance.
    Specifically, we prove the following.
    \begin{claim}
    Whenever there exists a set $\bfW \subset \bfX$ of size at most $F(s) \cdot s$ such that $\{\bfC_0 \setminus \bfW, \ldots, \bfC_m \setminus \bfW\}$ is an explainable clustering, there also exists a set $H \subset U$ that is a hitting set of $\mathcal{A}$ and $|H| \le |\bfW|$.
    On the other hand, for any hitting set $H \subset U$ there exists a solution $\bfW$ to the \probClustExpl instance such that $|\bfW| = |H|$.
    \label{claim:approximation}
    \end{claim}

    \begin{proof}
    In the forward direction, consider such a set $\bfW \subset \bfX$. We may assume that $\bfW \cap \bfC_j = \emptyset$ for each $j \in [m]$ and $\bfW \cap \bfO = \emptyset$, since these sets consist of $F(s) \cdot s + 1$ identical points, and replacing $\bfW$ by a smaller set not intersecing the sets above is still a solution. Thus, $\bfW \subset \bfV$. Recall that each of the points in $\bfV$ corresponds to an element of the universe $U$ in the \textsc{Hitting Set} instance. Denote by $H$ the subset of $U$ corresponding to $\bfW$, that is, $H = \{u \in U : \bfv_u \in \bfW\}$. Clearly $|H| \le |\bfW|$, we claim that $H$ is a solution to the \textsc{Hitting Set} instance.
   
    Consider a threshold tree that provides explanation for $\{\bfC_0 \setminus \bfW, \ldots, \bfC_m \setminus \bfW\}$, by the above all these clusters are non-empty. For $j \in \{0\}\cup [m]$, denote $\bfC_j' = \bfC_j \setminus \bfW$, recall that we assume $\bfC_j = \bfC_j'$ for all $j \in [m]$. We now show that for any $j \in [m]$, the set $H$ has a non-empty intersection with the set $S_j$.
    Since the tree represents the clustering, for each $j \in [m]$, there exists a cut in the tree that separates $\bfC_j'$ and $\bfC_0'$. That is, there exists a cut $\cut_{i, \theta} (\bfX') = (\bfX'_1, \bfX'_2)$ in the tree for $\bfX'  \subset \bfX$ such that $\bfC'_j \cup \bfC'_0 \subset \bfX'$, and $\bfC'_j \subset \bfX'_1$, $\bfC'_0 \subset \bfX'_2$ (w.l.o.g). The dimension of the cut $i$ necessarily belongs to $I_j$, since in all other coordinates $\bfC_j'$ is indistinguishable from $\bfO$. For $i \in I_j$, there exists a point $\bfv_u \in \bfV$ such that $\bfv_u[i] = 1$ and $u \in S_j$. Since the points in $\bfC_j'$ are set to one in this coordinate and the points in $\bfO$ to zero, $\bfv_u$ has to be in $\bfW$. Thus, by construction, $u \in H$, and the set $S_j$ is hit by $H$.

    In the other direction, assume that there exists a hitting set $H$ for $\mathcal{A}$.
    We set $\bfW$ to be the corresponding to $H$ subset of $\bfV$, $|\bfW| = |H|$. For every $j \in [m]$, there exists an element $u \in H$ such 
    that $u \in S_j$. Consider the corresponding element $\bfv_u$ of $\bfW$, and the coordinate $i \in I_j$ where $\bfv_u[i] = 1$.
    By construction,  $\bfv_u$ is the only element of $\bfC_0$ that has one in the $i$-th coordinate, and for any $j' \ne j$, the elements of $\bfC_{j'}$ have zero in thes coordinate. Thus, cutting over the dimension $i$ with the threshold $\theta$ set to zero separates $\bfC_j$ from all the other clusters. Finally, the threshold tree for $\{\bfC_0 \setminus \bfW, \ldots, \bfC_m \setminus \bfW\}$ is constructed by separating out each $\bfC_j$ one by one via the corresponding cut.
    \end{proof}
    
    The theorem follows easily from Claim~\ref{claim:approximation}. Namely, assume there exists an $F(s)$-approximate algorithm for \probClustExpl with running time $f(s) (nd)^{o(s)}$, that is, an algorithm that correctly decides either that the input instance has a solution of size at most $F(s) \cdot s$, or that it has no solution of size at most $s$. We construct an $F(\ell)$-approximate algorithm for \textsc{Hitting Set} as follows. First, construct a \probClustExpl instance via the reduction above. Second, run the \probClustExpl algorithm
    on that instance and output its answer. If it returns that there is a solution $\bfW$ of size at most $F(s) \cdot s$, by Claim~\ref{claim:approximation}
    there is a solution $H$ to the original \textsc{Hitting Set} instance of size at most $F(s) \cdot s = F(\ell) \cdot \ell$.
    If there is no solution to the \probClustExpl instance of size at most $s$, by Claim~\ref{claim:approximation} there also cannot be a solution to the \textsc{Hitting Set} instance of size at most $s = \ell$. This shows that the constructed algorithm provides indeed an $F(\ell)$-approximation for the \textsc{Hitting Set} problem. Finally, its running time can be bounded as $f(s) (nd)^{o(s)} = f(\ell)(F(\ell)|U| m)^{o(\ell)}$,
    which contradicts Theorem~\ref{thm:hitting_set_hardness} unless ETH fails.
\end{proof}

\section{Explainable Clustering}
\label{sec:explainable}
\paragraph{Explainable \probmeanmed{} clustering.} We consider the \probEMeanClust (resp. \probEMedClust) problem where given a collection $\bfX\subseteq \mathbb{R}^d$ of $n\geq k$ points, the task is to find an explainable $k$-clustering $\{\bfC_1,\ldots,\bfC_k\}$ of $\bfX$ of minimum $k$-means (resp. $k$-median) cost.
 %Similarly, in  \probEMedClust, the aim is to find  an explainable $k$-clustering minimizing the $k$-median cost. 

%Next, we describe our algorithmic and hardness results on Explainable \probmeanmed{} clustering. 

\subsection{Exact algorithms} 

Our $(nd)^{k+\Oh(1)}$ time algorithm is indeed very simple and based on branching technique. At each non-leaf node of threshold tree, we would like to find an optimal cut. As we focus on canonical threshold trees, the number of distinct choices for branching is at most $nd$. Also as the number of non-leaf nodes in the threshold binary tree is $k-1$, we have the following theorem. 

\begin{theorem}\label{obs:naive}
 \probEMeanClust and  \probEMedClust can be solved in $(nd)^{k+\Oh(1)}$ time.
 \end{theorem}

 Our $n^{2d}\cdot (dn)^{\Oh(1)}$ time algorithm is based on dynamic programming, which we describe in the following.  For two vectors $\bfx,\bfy\in\mathbb{R}^d$, we write $\bfx\leq \bfy$  ($\bfx< \bfy$, respectively) to denote that $\bfx[i]\leq \bfy[i]$ ($\bfx[i]< \bfy[i]$, respectively) for every $i\in\{1,\ldots,d\}$. We highlight that when we write $\bfx<\bfy$, we require the strict  inequality for \emph{every} coordinate.  
 
 \begin{theorem}\label{prop:DP}
  \probEMeanClust and  \probEMedClust can be solved in $n^{2d}\cdot (dn)^{\Oh(1)}$ time.
  \end{theorem}
   
 \begin{proof} 
 The algorithms for both problems are almost the same. Hence, we demonstrate it for \probEMeanClust. For simplicity, we only show how to find the minimum cost of clustering but the algorithm can be easily modified to produce an optimal clustering as well by standard arguments.
 
 Let $(\bfX,k)$ be an instance of the problem with $\bfX=\{\bfx_1,\ldots,\bfx_n\}$ and $\bfX\subseteq\mathbb{R}^d$. We say that a vector $\bfz\in(\mathbb{R}\cup\{\pm\infty\})^d$ is \emph{canonical}  if $\bfz[i]\in \coord_i(\bfX)\cup\{\pm\infty\}$ for every $i\in\{1,\ldots,d\}$. For every pair of canonical vectors $(\bfa,\bfb)$ such that $\bfa\leq \bfb$ and every positive integer $s\leq k$, we compute the minimum means cost of an explainable $s$-clustering of $\bfX_{\bfa,\bfb}=\{\bfx_i\in \bfX\mid \bfa< \bfx_i\leq\bfb\}$ and denote this value $\omega(\bfa,\bfb,s)$.  We assume that $\omega(\bfa,\bfb,s)=+\infty$ if $\bfX_{\bfa,\bfb}$ does not admit an explainable $s$-clustering.  It is also convenient to assume that  $\omega(\bfa,\bfb,s)=+\infty$ if $\bfX_{\bfa,\bfb}=\emptyset$, because we are not interested in empty clusters.
  Notice that the minimum means cost of an explainable $k$-clustering of $\bfX$ is $\omega(\bfa^*,\bfb^*,k)$, where $\bfa^*[i]=-\infty$ and $\bfb^*[i]=+\infty$ for $i\in\{1,\ldots,d\}$. We compute the table of values of $\omega(\bfa,\bfb,s)$ consecutively for $s=1,2,\ldots,k$.
 
 If $s=1$, then by definition, 
 $$\omega(\bfa,\bfb,s)=
 \begin{cases}
 \cost_2(\bfX_{\bfa,\bfb})&\mbox{if }\bfX_{\bfa,\bfb}\neq\emptyset,\\
 +\infty&\mbox{if }\bfX_{\bfa,\bfb}=\emptyset, 
 \end{cases}
 $$ 
 and this value can be computed in polynomial time.  Let $s\geq 2$ and assume that the tables are already constructed for the lesser values of $s$. Consider a pair $(\bfa,\bfb)$ of canonical vectors of 
 $(\mathbb{R}
 \cup\{\pm\infty\})^d$ such that $\bfa\leq \bfb$. For $i\in \{1,\ldots,d\}$ and $\theta\in \coord_i(\bfX)$ such that $\bfa[i]< \theta<\bfb[i]$, we define the vectors $\bfa^{i,\theta}$ and $\bfb^{i,\theta}$ by setting
 $$\bfa^{i,\theta}[j]=
 \begin{cases}
 \theta&\mbox{if }j=i,\\
 \bfa[j]&\mbox{if }j\neq i,
 \end{cases}
 \text{ and }
\bfb^{i,\theta}[j]=
 \begin{cases}
 \theta&\mbox{if }j=i,\\
 \bfb[j]&\mbox{if }j\neq i.
 \end{cases}
$$
Then we compute $\omega(\bfa,\bfb,s)$ using the following recurrence
\begin{align}\label{eq:rec-dp}
\omega(\bfa,\bfb,s)=\min\{\omega(\bfa,&\bfb^{i,\theta},s_1)+\omega(\bfa^{i,\theta},\bfb,s_2)\nonumber\\
&\text{for }  1\leq i\leq d, \theta\in\coord_i(\bfX),\nonumber\\
& \bfa[i]< \theta< \bfb[i],\nonumber\\
& s_1,s_2\geq 1,\text{ and }s_1+s_2=s\}.
\end{align}
 The correctness of (\ref{eq:rec-dp}) follows form the definition of the explainable clustering. It is sufficient to observe that to compute the optimum means cost of an explainable $s$-clustering of $\bfX_{\bfa,\bfb}$, we have to take minimum 
 over the sums of  optimum costs of explainable  $s_1$-clusterings and $s_2$-clusterings of $X_1$ and $X_2$, respectively, where $(X_1,X_2)=\cut_{i,\theta}(\bfX_{\bfa,\bfb})$ for some $i\in \{1,\ldots,d\}$, $\theta\in\coord_i(\bfX_{\bfa,\bfb})$ and $s_1+s_2=s$, and this is exactly what is done in (\ref{eq:rec-dp}).  
 
To evaluate the running time, observe that to compute $\omega(\bfa,\bfb,s)$ using (\ref{eq:rec-dp}), we consider $d$ values of $i$, at most $n$ values of $\theta$ and at most $k\leq n$ values of $s_1$ and $s_2$, that is, we go over at most  $dn^2$ choices. Thus, computing $\omega(\bfa,\bfb,s)$ for $s\geq 2$ and fixed $\bfa$ and $\bfb$  can be done in $\Oh(dn^2)$ time. Since there are at most $(n+2)^{2d}$ pairs of canonical vectors $\bfa$ and  $\bfb$, we obtain that the  time to compute the table of values of $\bfX_{\bfa,\bfb}$ for all pairs of vectors is $n^{2d}\cdot (dn)^{\Oh(1)}$. Since the table for $s=1$ can be constructed in $n^{2d}\cdot (dn)^{\Oh(1)}$ and we iterate using (\ref{eq:rec-dp}) $k-1\leq n$ times, the total running time is $n^{2d}\cdot (dn)^{\Oh(1)}$.
 \end{proof} 

% %\paragraph{Parameterized complexity and ETH.} We refer to the recent books~\cite{CyganFKLMPPS15,FominLSZ19} for the formal introduction to the area of parameterized complexity. Our hardness results are based on the \emph{Exponential Time Hypothesis (ETH)} formulated by Impagliazzo, Paturi and Zane~\cite{ImpagliazzoP99,ImpagliazzoPZ01}. For an integer $k\geq 3$, let $q_k$ be the infimum of the real numbers $c$ such that the \textsc{$k$-Satisfiability} problem can be solved in time $\Oh(2^{c n})$, where $n$ is the number of variables. Exponential Time Hypothesis states that $\delta_3>3$. In particular, ETH implies that  \textsc{$k$-Satisfiability} cannot be solved in time $2^{o(n)}n^{\Oh(1)}$.

%Next, we briefly describe the hardness result. 
%!TEX root = Explainable clustering.tex
%\section{Hardness of Explainable Clustering} 
\subsection{Hardness}
\label{sec:hardexplncl}

In this section, we show our hardness results for \probEMeanClust and \probEMedClust: the problems are  \classNP-complete, \classW{2}-hard when parameterized by $k$, and cannot be solved in $f(k)\cdot n^{o(k)}$ time for a computable function $f(\cdot)$ unless ETH fails.
Moreover, the hardness holds even if the input points are binary. More precisely, we prove the following theorem.

%To prove this theorem, we again reduce from \textsc{Hitting Set}, but the construction is different. Here, we construct a point for each set and also for each element. Then, there is a hitting set of size $k$ iff there is an explainable $(k+1)$-clustering of a suitable cost.

\begin{theorem}
\label{thm:exp-ethhard} 
Given a collection of $n$ points $\bfX\subseteq \{0,1\}^d$, a positive integer $k\leq n$, and a nonnengative integer $B$, it is \classW{2}-hard to decide whether $\bfX$ admits an explainable $k$-clustering of  mean (median, respectively) cost at most $B$, when the problem is parameterized by $k$. Moreover, the problems are  \classNP-complete and cannot be solved in $f(k)\cdot n^{o(k)}$ time for a computable function $f(\cdot)$ unless ETH fails.
\end{theorem}

\begin{proof}
We show the theorem for the means costs and then briefly explain how the proof is modified for medians. The case of median cost is easier due to the fact that for a collection of binary points, its median is also a binary vector. 

We reduce from the \probHS problem. The task of the problem is, given a family of sets $\mathcal{W}=\{W_1,\ldots,W_m\}$ over  universe $U=\{u_1,\ldots,u_n\}$, and a positive integer $k$, decide whether there  is  $S\subseteq U$ of size at most $k$ that is a  \emph{hitting} set for $\mathcal{W}$, i.e.  such that  $S\cap W_i\neq\emptyset$ for every $i\in\{1,\ldots,m\}$. This problem is well-known to be \classW{2}-complete  when parameterized by $k$ even if the input is restricted to families of sets of the same size (see~\cite{DowneyF13}). 

Let $\mathcal{W}=\{W_1,\ldots,W_m\}$ be a family of sets over $U=\{u_1,\ldots,u_n\}$ with $|W_1|=\cdots=|W_m|=r$, and let $k$ be a positive integer. We construct the following points.
\begin{itemize}
\item Let $s=8nm^2$. For each $i\in\{1,\ldots,n\}$, we construct a vector $\bfu_i\in\{0,1\}^n$ by setting 
$$\bfu_i[j]=
\begin{cases}
1&\mbox{if }j=i,\\
0&\mbox{otherwise}.
\end{cases}
$$
We also define $\bfU_i$ to be the collection of $s$ points identical to $\bfu_i$.
\item For each $i\in\{1,\ldots,m\}$, let $\bfw_i$ be the characteristic vector of $W_i$.
That is, 
$$\bfw_i[j]=
\begin{cases}
1&\mbox{if }u_j\in W_i,\\
0&\mbox{otherwise},
\end{cases}
$$
 We define $\bfW=\{\bfw_1,\ldots,\bfw_m\}$.
\item For $t=16ns^2$ we construct a collection $\bfZ$ of $t$ zero points.
\end{itemize}
Finally,  we define
 \begin{itemize}
 \item $\bfX=\big(\bigcup_{i=1}^n\bfU_i\big)\cup\bfW\cup\bfZ$,
 \item $k'=k+1$, and
 \item $B=m(r-1)+s(n-k)$.
 \end{itemize}
 
 Clearly, $\bfX$, $k'$ and $B$ can be constructed form the considered instance of \probHS in polynomial time.
 We claim that $\mathcal{W}$ has a hitting set $S$ of size at most $k$ if and only if $\bfX$ has an explainable $k'$-clustering of means cost at most $B$.
 
 For the forward direction, assume that $S$ is a hitting set of size at most $k$. Without loss of generality we assume that $|S|=k$. Let $S=\{u_{i_1},\ldots,u_{i_k}\}$. We define the $k'$-clustering $\{\bfC_1,\ldots,\bfC_{k+1}\}$ as follows. We set $\bfC_1=\{\bfx\in\bfX\mid \bfx[i_1]=1\}$ and for $j=2,\ldots,k+1$, 
 $$\bfC_j=\{\bfx\in\bfX\mid \bfx[i_j]=1\}\setminus\bigcup_{h=1}^{j-1}\bfC_h.$$ 
 Observe that the constructed clustering is explainable. To see this, let $X_0=\bfX$ and set $X_j=\bfX\setminus \bigcup_{h=1}^{j-1}\bfC_h$. Then $(X_j,\bfC_j)=\cut_{i_j,0}(X_{j-1})$ for $j\in \{1,\ldots,k+1\}$. Notice that $\bfU_{i_j}\subseteq \bfC_j$ for every $j\in\{1,\ldots,k\}$, and we have that $\bfZ\subseteq \bfC_{k+1}$ and $\bfU_h\subseteq \bfC_{k+1}$ for all $h\in\{1,\ldots,m\}\setminus \{i_1,\ldots,i_k\}$. Moreover, because $S$ is a hitting set, each $\bfw_j\in\bfW$ is included in some cluster $\bfC_h$ for some $h\in\{1,\ldots,k\}$, that is, $\bfC_{k+1}=\bfZ\cup \bfR$ , where $\bfR= \bigcup_{h\in\{1,\ldots,m\}\setminus \{i_1,\ldots,i_k\}}\bfU_h$. 
 
 For every $j\in\{1,\ldots,k\}$, we define $\bfc_j=\bfu_{i_j}$, and we set $\bfc_{k+1}$ be the zero vector. Clearly, for every $j\in\{1,\ldots,k+1\}$, $\cost_2(\bfC_j)\leq \sum_{\bfx\in \bfC_j}\|\bfx-\bfc_j||_2^2$. Let $\bfW_j=\bfW\cap\bfC_j$ for $j\in\{1,\ldots,k\}$. Note that $\bfC_j=\bfW_j\cup\bfU_{i_j}$ for $j\in\{1,\ldots,k\}$.
 Because $\bfx[i_j]=1$ for every $\bfx\in\bfC_j$ for  $j\in \{1,\ldots,k\}$ and $\bfx$ has exactly $r$ nonzero elements, we have that $\|\bfx-\bfc_j\|_2^2=r-1$ for every $\bfx\in\bfW_j$ for $j\in\{1,\ldots,k\}$. 
 If $\bfx\in \bfR$, then $\|\bfx-\bfc_{k+1}\|_2^2=1$. We obtain that 
 \begin{multline*}
 \cost_2(\bfC_1,\ldots,\bfC_{k+1})\leq \sum_{j=1}^{k+1}\sum_{\bfx\in \bfC_j}\|\bfx-\bfc_j\|_2^2
 =\sum_{j=1}^k\sum_{\bfx\in\bfW_j}\|\bfx-\bfc_j\|_2^2+\sum_{\bfx\in\bfR}\|\bfx-\bfc_{k+1}\|_2^2\\
 =m(r-1)+s(n-k)=B.
 \end{multline*}
 Thus, $\{\bfC_1,\ldots,\bfC_{k+1}\}$ is an explainable $k'$-clustering of means cost at most $B$. 
 
 For the opposite direction, let $\{\bfC_1,\ldots,\bfC_{k+1}\}$ be an explainable $k'$-clustering with \linebreak $\cost_2(\bfC_1,\ldots,\bfC_{k+1})\leq B$. Let also $(T,k',\varphi)$ be a canonical threshold tree for this clustering.  Notice that because every point of $\bfX$ is binary, for every nonleaf $v\in V(T)$, $\varphi(v)=(i,0)$ for some $i\in\{1,\ldots,d\}$.

We show the following property of  $(T,k',\varphi)$: for every nonleaf node $v\in V(T)$, its right child is a leaf. Suppose that this is not the case. Denote by $x$ the root of $T$ and let $P=x_1\cdots x_p$ be the root-leaf path, where $x_1=x$ and $x_i$ is the left child of $x_{i-1}$ for $i\in\{2,\ldots,p\}$, that is, $P$ is constructed starting from the root and following the left children until we achieve the leftmost leaf.  Because  $T$ has $k'$ leaves and at least one right  child is not a leaf, we have that $p\leq k'-1=k$. Denote by $\bfC_q$ the cluster corresponding to the leaf $x_p$ of $T$. Notice that $\bfZ\subseteq \bfC_q$ and $n-(p-1)$ collections $\bfU_i$ are included in $\bfC_q$ for some $i\in\{1,\ldots,n\}$.
 Let $i_1,\ldots,i_{n-p+1}\in\{1,\ldots,n\}$ be the distinct indices such that $\bfU_{i_j}\subseteq \bfC_q$.   Denote by $\bfc$ the mean of $\bfC_q$. Observe that for each $j\in\{1,\ldots,n\}$, the multiset of the $j$-th coordinates of the points of $\bfC_q$ contains at most $s+m$ ones and at least $t$ zeros.  Then for every $j\in \{1,\ldots,n\}$,
 $$\bfc[j]\leq \frac{m+s}{|\bfC_q|}\leq \frac{m+s}{t}\leq\frac{2s}{t}=\frac{1}{8ns}\leq \frac{1}{2s},$$
 and we obtain that 
 \begin{align*}
 \cost_2(\bfC_q)\geq& \sum_{j=1}^{n-p+1}\sum_{\bfx\in \bfU_{i_j}}\|\bfx-\bfc\|_2^2=s\sum_{j=1}^{n-p+1}\|\bfu_{j_i}-\bfc\|_2^2 \\
 \geq&s(n-k+1)\big(1-\frac{1}{2s}\big)^2> s(n-k+1)\big(1-\frac{1}{s}\big)\\ 
 \geq &s(n-k)+s-n \geq s(n-k)+2nm-n\\
 >&m(n-1)+s(n-k)\geq B,
 \end{align*}
 contradicting that  $\cost_2(\bfC_1,\ldots,\bfC_{k+1})\leq B$. 
 
 Because for every nonleaf node $v\in V(T)$, its right child is a leaf,  we have that the clustering is obtained by consecutive cutting each cluster from the set of points.  
 Then  we can assume without loss of generality that there a $k$-tuple of distinct indices $(i_1,\ldots,i_k)$ from $\{1,\ldots,n\}$ such that 
  $\bfC_1=\{\bfx\in\bfX\mid \bfx[i_1]=1\}$ and  
 $$\bfC_j=\{\bfx\in\bfX\mid \bfx[i_j]=1\}\setminus\bigcup_{h=1}^{j-1}\bfC_h$$ 
 for $j=2,\ldots,k+1$. We claim that $S=\{u_{i_1},\ldots,u_{i_k}\}$ is a hitting set for $\mathcal{W}$.
 
The proof is by contradiction. Suppose that $S$ is not a hitting set. 
%Denote by $I\subseteq \{1,\ldots,m\}$ the subset of indices such that $W_i\cap S\neq\emptyset$ for $i\in I$, and let $J=\{1,\ldots,m\}\setminus I$. Clearly, $|I|+|J|=m$ and, because $S$ is not a hitting set, $|J|\geq1$. 
For every $i\in\{1,\ldots,k+1\}$, let $\bfW_i=\bfC_i\cap \bfW$. Notice that because $S$ is not a hitting set, $\bfW_{k+1}\neq\emptyset$. 
We analyse the structure of clusters and upper bound their means costs. For this, denote by $\bfc_1,\ldots,\bfc_{k+1}$ the means of $\bfC_1,\ldots,\bfC_{k+1}$. 

Let $j\in \{1,\ldots,k\}$ and consider $\bfC_j$ with the mean $\bfc_j$. We have that $\bfC_j=\bfU_{i_j}\cup W_j$. Then $\bfc_j[i_j]=1$. Let $h\in \{1,\ldots,n\}$ be distinct from $i_j$. If $\bfx\in \bfU_{i_j}$, then $\bfx[h]=0$.
%, and $\bfx[h]=1$ for $\bfx\in \bfU_{i_j}$. 
Then the multiset of the $h$-th coordinates of the points of $\bfC_j$ contains at most $|\bfW_j|\leq m$
ones and at least  $s$ zeros, and we have that  
 $$\bfc_j[h]\leq \frac{m}{|\bfC_j|}\leq \frac{m}{s}\leq\frac{1}{8mn}.$$ 
 Recall that $\bfx\in \bfW$ has exactly $r$ elements that are equal to one. This implies, that for every $\bfx\in \bfW_j$,
 \begin{equation*}
 \|\bfx-\bfc_j\|_2^2\geq (r-1)\big(1-\frac{1}{8mn}\big)^2\geq (r-1)\big(1-\frac{1}{4mn}\big)
 \end{equation*}
 and, therefore,
  \begin{equation}\label{eq:first-clusters}
 \cost_2(\bfC_j)\geq (r-1)\big(1-\frac{1}{4mn}\big)|\bfW_j|
 \geq (r-1)|\bfW_j|-\frac{1}{4m}|\bfW_j|,
 \end{equation} 
 because $r\leq n$.

 Now we consider $\bfC_{k+1}$ and the corresponding mean $\bfc_{k+1}$. Notice that $\bfC_{k+1}=\bfZ\cup \bfW_{k+1}\cup \bfR$, where $\bfR= \bigcup_{h\in\{1,\ldots,m\}\setminus \{i_1,\ldots,i_k\}}\bfU_h$.
 Let $h\in \{1,\ldots,n\}$. We have that $\bfx[h]=0$ for $\bfx\in \bfZ$. Hence, the multiset of of $h$-th coordinates of the points of $\bfC_{k+1}$ contains at most $|\bfW_{k+1}|+s\leq m+s$
ones and at least $t$ zeros. Therefore,   
 $$\bfc_{k+1}[h]\leq \frac{m+s}{|\bfC_{k+1}|}\leq \frac{2s}{t}\leq\frac{1}{8sn}.$$ 
Since $\bfx\in \bfW$ has exactly $r\leq n$ elements that are equal to one,  for every $\bfx\in \bfW_{k+1}$,
\begin{equation*}
\|\bfx-\bfc_{k+1}\|_2^2\geq  r\big(1-\frac{1}{8sn}\big)^2 \geq r\big(1-\frac{1}{8mn}\big)^2
\geq r\big(1-\frac{1}{4mn}\big)\geq r-\frac{1}{4m}.
\end{equation*}
For every $\bfx\in \bfR$,  $\bfx$ contains a unique nonzero element and, therefore, 
\begin{equation*}
\|\bfx-\bfc_{k+1}\|_2^2\geq \big(1-\frac{1}{8sn}\big)^2\geq 1-\frac{1}{4sn}.
\end{equation*}
Note that $\bfR$ contains exactly $s(n-k)$ points.
Then 
\begin{multline}\label{eq:last-cluster}
\cost_2(\bfC_{k+1})\geq\sum_{\bfx\in \bfW_{k+1}}\|\bfx-\bfc_{k+1}\|_2^2+\sum_{\bfx\in \bfR}\|\bfx-\bfc_{k+1}\|_2^2\\
\geq \big(r-\frac{1}{4m}\big)|\bfW_{k+1}|+s(n-k)\big(1-\frac{1}{4sn}\big)\\
\geq r|\bfW_{k+1}|+s(n-k)-\frac{1}{4m}|\bfW_{k+1}|-\frac{1}{4}.
\end{multline}

Recall that $|\bfW_1|+\cdots+|\bfW_{k+1}|=m$. Then combining (\ref{eq:first-clusters}) and (\ref{eq:last-cluster}), we obtain that 
\begin{equation*}
 \cost_2(\bfC_1,\ldots,\bfC_{k+1})=\sum_{j=1}^{k+1}\cost_2(\bfC_j)
 \geq m(r-1)+s(n-k)+|\bfW_{k+1}|-\frac{1}{2}
 \geq B+\frac{1}{2},
\end{equation*}
because $\bfW_{k+1}\neq\emptyset$.
However, this contradicts that the means cost of $\{\bfC_1,\ldots,\bfC_{k'}\}$ is at most $B$. This means that $S$ is a hitting set for $\mathcal{W}$. 

This concludes the hardness proof for \probEMeanClust. For the median cost, we use exactly the same reduction from \probHS. Then we show that $\mathcal{W}$ has a hitting set $S$ of size at most $k$ if and only if $\bfX$ has an explainable $k'$-clustering of median cost at most $B$. The proof for the forward direction is identical to the proof for the means up to the replacement of $\cost_2$ by $\cost_1$ and of $\|\cdot\|_2^2$ by $\|\cdot\|_1$. For the opposite direction, 
the proof follows the same lines as the above proof for means but gets simplified, because we can assume that medians are  binary. In particular, if the multiset of $h$-th coordinates of the points of a cluster $\bfC_{j}$ contains 
more zeros that ones, then $\bfc_j[h]=0$ for the median $\bfc_j$. Notice that the crucial part of the proof for the means is obtaining upper bounds for the values $\bfc_j[h]$. Now we can immediately assume that $\bfc_j[h]=0$ in all considered cases whenever we upper bound $\bfc_j[h]$, and the lower bounds for the costs in the proof become straightforward. 

To see the second part of the theorem, note that our reduction from  \probHS is polynomial. This immediately implies \classNP-hardness of the considered problems for both measures. For the lower bound up to ETH, we use the well-known fact (see, e.g.~\cite[Chapter~14]{CyganFKLMPPS15}) that \probHS does not admit an algorithm with running time $f(k)\cdot (n+m)^{o(k)}$ for any computable function $f(\cdot)$ up ETH. Because our reduction is polynomial and, moreover, the value of the parameter in the constructed instance is $k'=k+1$, we obtain that an algorithm with running time $f(k)\cdot (n+m)^{o(k)}$ for our problems would contradict ETH.
\end{proof}

 %!TEX root = neurips_2021.tex
\section{Approximate Explainable Clustering}
\label{sec:fractional}

\paragraph{Approximate explainable \probmeanmed{} clustering.} In \probFMeanExpl, we are given a collection of $n$ points $\bfX\subseteq \mathbb{R}^d$, a positive integer $k\leq n$, and a positive real constant $\varepsilon<1$. Then the task is to find a collection of points $\bfY\subseteq \bfX$ with $|\bfY|\geq (1-\varepsilon)|\bfX|$ and an explainable $k$-clustering of $\bfY$ whose $k$-median cost does not exceed the optimum $k$-median cost of an explainable $k$-clustering for the original collection of points $\bfX$. Note that we ask about the construction of $\bfY$ and the corresponding clustering as the decision variant is trivial. Observe also that the optimum cost is unknown a priori. 
 \probFMedExpl differs only by the clustering measure.

% In this section, we study   \probFMeanExpl and \probFMedExpl. Let us remind that in these problems, we are given a collection of $n$ points $\bfX\subseteq \mathbb{R}^d$, a positive integer $k\leq n$, and a positive real constant $\varepsilon<1$. Then the task is to find a collection of points $\bfY\subseteq \bfX$ with $|\bfY|\geq (1-\varepsilon)|\bfX|$ and an explainable $k$-clustering of $\bfY$ whose median or mean cost, respectively, does not exceed the optimum cost of an explainable $k$-clustering for the original collection of points $\bfX$. More precisely, we prove the following result. 

\begin{theorem}\label{thm:frac}
    \probFMeanExpl and \probFMedExpl are solvable in $(\frac{8dk}{\epsilon})^k\cdot n^{\Oh(1)}$ time. 
\end{theorem}

As the proofs for %\probFMeanExpl and \probFMedExpl 
both problems are identical, we describe only the algorithm for \probFMeanExpl. 

\begin{proof}

Let $\bfX\subseteq \mathbb{R}^d$ be an instance of \probFMeanExpl, $(T,k,\varphi)$ be the optimal (canonical) threshold tree for explainable $k$-means clustering and $(C_1, \dots, C_k)$ the clustering induced by $(T,k,\varphi)$. The goal of the algorithm will be to guess an approximation of $(T,k,\varphi)$. Since $T$ is a binary tree with $k$ leaves, guessing $T$ only requires $4^{k}$ tries. Guessing $\varphi$ is more complicated however, as there is $d \cdot n$ choice at each node of $T$, which gives potentially $(dn)^k$ possibilities, where $n = |\bfX|$. The idea here will be to guess for every nonleaf node $u$ of $T$ the second element of $\varphi(u)$ up to a precision of $\Oh(\frac{\epsilon n}{k})$, which gives only $\Oh(d \cdot \frac{k}{\epsilon})$ choices at each nonleaf node. 

More formally, let $n' = \lfloor \frac{\epsilon n}{k} \rfloor$ and note first that if $n' = 0$, then $\frac{\epsilon n}{k} < 1$ and thus $n \leq \frac{k}{\epsilon}$. This means that if $n' = 0$, then the algorithm  trying all the possible values of $T$ and $\varphi$, and computing the value of the obtained clustering, runs in time $4^k(\frac{dk}{\epsilon}) \cdot n^{\Oh(1)}$, which ends the proof. From now on, let us assume that $n' \not = 0$. 

Let $U$ denote the set of nonleaf nodes of $T$. Let  $\varphi'\colon U\rightarrow \{1,\ldots,d\}\times \mathbb{R}$ be the function obtained from $\varphi$ by rounding, for every $u \in U$, the value of the first element of $\varphi(u)$ to the closest multiple of $n'$. In other words, if $\varphi(u) = (j,r)$, then $\varphi'(u) = (j,i \cdot n')$ where $i$ is the largest integer such that $i \cdot n' \leq r$. 

Consider now the clustering obtained from the threshold tree $(T,k,\varphi')$. At each node $v \in T$ such that $\varphi(v) = (j,r)$ and $\varphi'(v) = (j,i \cdot n')$, the points $x$ of $\bfX$ that can be misplaced by the $\cut_{\varphi'(u)}(\bfX)$ are exactly the points such that $i \cdot n' < \bfx[i] \leq r$. This means that, if $\bfZ_u$ denotes the set of all the points $x$ such that $i \cdot n' \leq \bfx[j] \leq (i+1)n'$, then $|\bfZ_u| \leq n'$ and the partitions $\cut_{\varphi'(u)}(\bfX \setminus \bfZ_u)$ and $\cut_{\varphi(u)}(\bfX \setminus \bfZ_u)$ are identical. Therefore, if $\bfZ = \bigcup_{u \in U} \bfZ_u$, then $(T,k,\varphi')$ and $(T,k,\varphi)$ induce the exact same clustering on $\bfX \setminus \bfZ$. Note that $|\bfZ| \leq kn' \leq \epsilon n$. 

Therefore the algorithm will try all possible choices for $T$, and for every nonleaf node $u$, it tries all possible values for $\varphi'(u)$ of the form $(j, i \cdot n')$, where $j \in [d]$ and $i \in [\frac{2k}{\epsilon}]$. For each such try, the algorithm also removes the set $\bf Z$ consisting  of all the points $x$ such that $  i \cdot n' \leq \bfx[j] \leq (i+1) \cdot n'$ whenever there exists $u \in U$ such that $\varphi'(u) = (j, i \cdot n')$ and computes the value of the clustering induces by $(T,k, \varphi')$ on $\bfX \setminus \bfZ$. Finally it outputs the set $\bf X \setminus \bfZ$ as well as the threshold tree $(T,k, \varphi')$ which minimises the value of the clustering. 

Note that for every set of choices of $T$ and $\varphi'(u)$ of the form $(j, i \cdot n')$, the set $\bfZ$ has size at most $k \cdot n' \leq \epsilon n$, which implies that the algorithm indeed outputs the desired set and threshold tree. Moreover, since there are at most $4^{k}$ possible trees $T$ and $ d \cdot \frac{2 k}{\epsilon}$ possible choices of $\varphi'(u)$ for every node of the tree, we conclude that the  running time of the algorithm is $(\frac{8dk}{\epsilon})^kn^{\Oh(1)}$.
\end{proof}

%!TEX root = neurips_2021.tex
\section{Conclusion}\label{sec:concl}
In this paper, we initiated the study of computational complexity of several variants of explainable clustering. Concluding, we would like to outline some further directions of research and state a number of open problems. 

We showed that \probClustExpl admits a polynomial-time approximation with a factor of $(k-1)$. Can this factor be improved in polynomial-time, say, to $\log k$? We proved that \probEMeanClust and  \probEMedClust can be solved in $n^{2d}\cdot (dn)^{\Oh(1)}$ time. Is this result tight? Or is it possible to obtain an $f(d)\cdot (dn)^{\Oh(1)}$ time algorithm for some function $f$? Also, is it possible to obtain approximation schemes parameterized by $k$, i.e., $(1+\varepsilon)$-approximation in $g(k,\varepsilon) (nd)^{\Oh(1)}$ time for some function $g$. Regular $k$-means/median admits such approximation schemes \cite{KumarSS10}. 

\bibliographystyle{plainurl}
\bibliography{Clustering-new}   
%\appendix
%\clearpage
%\input{appendix}
\end{document}